\documentclass[12pt]{article}
\usepackage{amsmath}
\usepackage{amsfonts}
\usepackage{amssymb}
\usepackage{amsthm}
\usepackage{graphicx}
\usepackage{fullpage}
\usepackage{setspace}
\usepackage{verbatim}
\usepackage{natbib}
\usepackage{appendix}
\usepackage{rotating}
\newtheorem{theorem}{Theorem}[section]
\newtheorem{lemma}{Lemma}[section]

\newtheorem{assumption}{Assumption}[section]

\title{On the Choice of Test Statistic for Conditional Moment
  Inequalities}

\author{Timothy B. Armstrong\thanks{email: timothy.armstrong@yale.edu.  Support from National Science Foundation
Grant SES-1628939 is gratefully acknowledged.}
\\
Yale University}

\begin{document}

\maketitle

\begin{abstract}
This paper derives
asymptotic approximations to the power of
Cramer-von Mises (CvM) style tests 
for inference on a finite dimensional parameter defined by conditional moment inequalities in the case where the parameter is set identified.
Combined with power results for Kolmogorov-Smirnov (KS) tests, these results can be used to choose the optimal test statistic, weighting function and, for tests based on kernel estimates, kernel bandwidth.
The results show that, in the setting considered here, KS tests are preferred to CvM tests, and that a truncated variance weighting is preferred to bounded weightings.
\end{abstract}

\section{Introduction}\label{introduction_sec}

This paper compares methods for inference on a parameter $\theta$ defined by the conditional moment inequalities
\begin{align*}%
E(m(W_i,\theta)|X_i)\ge 0 \,\,\, a.s.
\end{align*}
where $m:\mathbb{R}^{d_W+d_\theta}\to \mathbb{R}^{d_Y}$ is a known
function of data $W_i$ and a parameter $\theta\in\Theta\subseteq \mathbb{R}^{d_\theta}$, and $\ge$ is
defined elementwise.
Here, $W_i$ is a $\mathbb{R}^{d_W}$ valued random variable and $X_i$ is a $\mathbb{R}^{d_X}$ valued random variable.  We are given independent, identically distributed (iid) observations $\{(X_i',W_i')'\}_{i=1}^n$.
This defines the identified set
\begin{align*}
\Theta_0\equiv \{\theta\in\Theta|
   E(m(W_i,\theta)|X_i)\ge 0 \,\,\, a.s.
 \}
\end{align*}
where $\Theta\subseteq\mathbb{R}^{d_\theta}$ is the parameter space.
If $\Theta_0$ contains more than one element, the model is said to be
set identified.

Following \citet{imbens_confidence_2004}, we are interested in confidence regions $\mathcal{C}_n$ that satisfy the converage criterion
\begin{align}\label{im_coverage_eq}
\text{for all } \theta_0\in\Theta_0, \,\liminf_{n\to\infty} P(\theta_0\in \mathcal{C}_n)\ge 1-\alpha.
\end{align}
We consider confidence regions constructed by inverting
a family of tests $\phi_n(\theta)=\phi_n(\theta,\{X_i,W_i\}_{i=1}^n)$, where $\phi_n(\theta)$ is a test of $H_{0,\theta}:\theta\in\Theta_0$:
\begin{align*}
\mathcal{C}_n=\{\theta|\phi_n(\theta)=0\}.
\end{align*}
Subject to the coverage criterion (\ref{im_coverage_eq}), we would like the confidence region $\mathcal{C}_n$ not to contain points that are far away from the identified set $\Theta_0$.  In particular, if we take a parameter $\theta_0$ on the boundary of $\Theta_0$ and consider a sequence $\theta_n=\theta_0+a_n$ where $a_n\to 0$, we would like to have $\theta_n\notin\mathcal{C}_n$ with high probability for $a_n$ converging to zero as quickly as possible (so long as $\theta_n$ approaches $\Theta_0$ from the outside, rather than from the interior).  Note that
\begin{align*}
P(\theta_n\notin\mathcal{C}_n)=P(\phi_n(\theta_n)=1).
\end{align*}
Thus, %
we can determine whether $\mathcal{C}_n$ contains points that are far away from $\Theta_0$
by examining the behavior of $P(\phi_n(\theta_n)=1)$, which is the power of the test $\phi_n(\theta_n)$ of $H_{0,\theta_n}$ at the alternative $P$.

This paper provides an asymptotic answer to this question by examining the asymptotic behavior of $P(\phi_n(\theta_n)=1)$ as $n\to\infty$.  We refer to limit of $P(\phi_n(\theta_n)=1)$ as the local asymptotic power of the sequence of tests $\phi_n(\theta_n)$ (note that this terminology differs from definitions often used in the literature, since the null hypothesis varies with $n$ while the alternative stays fixed).  The local asymptotic power of this sequence of tests will depend on the distribution $P$, the parameter $\theta_0$ on the boundary of $\Theta_0$ to which the sequence $\theta_n=\theta_0+a_n$ converges, and the sequence $a_n$.

This paper considers Cramer-von Mises (CvM) style test statistics, which integrate or add some function of the negative part of an objective function.  These can be compared with existing results for Kolmogorov-Smirnov (KS) statistics, which take the minimum of an objective function.  The results show that the power $P(\phi_n(\theta_n)=1)$ will be greater asymptotically for KS statistics when the distribution $P$ satisfies generic smoothness conditions of the form used in the nonparametric statistics literature.  In particular, the results imply that KS statistics are preferred according to a ``minimax within a smoothness class'' criterion of the form used to formulate nonparametric relative efficiency results in papers such as \citet{stone_optimal_1982}.

As an example of the types of problems covered by this setup, consider the interval regression model of \citet{manski_inference_2002}.  We observe $(X_i,W_i^L,W_i^H)$ where $[W_i^L,W_i^H]$ is known to contain the latent variable $W_i^*$, which follows the linear regression model $E(W_i^*|X_i)=(1,X_i')\theta$.  This falls into the setup of this paper with $W_i=(X_i,W_i^L,W_i^H)$ and $m(W_i,\theta)=(W_i^H-(1,X_i')\theta,(1,X_i')\theta-W_i^L)'$.  The identified set is then given by
\begin{align*}
\Theta_0=\{\theta|E(W_i^L|X_i)\le (1,X_i')\theta\le E(W_i^H|X_i) \,\,\, a.s.\}.
\end{align*}
Thus, a parameter $\theta_0$ in the identified set corresponds to a regression line $(1,x')\theta_0$ that is between the conditional means $E(W_i^L|X_i=x)$ and $E(W_i^H|X_i=x)$ for all $x$ on the support of $X_i$.  If $\theta_0$ is on the boundary of the identified set, it will be equal to one of these regression lines for some value of $x$.  For $\theta_n=\theta_0+a_n$ approaching the boundary of the identified set from the outside, the regression line $(1,x')\theta_n$ will be above $E(W_i^H|X_i=x)$ or below $E(W_i^L|X_i=x)$ for some values of $x$, and we would like the test $\phi_n(\theta_n)$ to detect this so that $\theta_n\notin \mathcal{C}_n$ with high probability.  We use primitive conditions to apply the general results in this paper to this setting, thereby giving asymptotic approximations to this probability.
These conditions correspond to smoothness conditions used in the nonparametric statistics literature and conditions on the shape of these conditional means near points where one of them is equal to $(1,x')\theta_0$ (see Section \ref{rootn_comparison_sec} and Appendix \ref{minimax_rates_sec}).

The remainder of this paper is organized as follows.
Section \ref{test_def_sec} defines the tests considered in this paper.
Section \ref{related_lit_sec} discusses related literature.
Section \ref{intuition_sec} gives an intuitive description of the power results in this paper and how they are derived.
Section \ref{larger_assumptions_sec} states formally the conditions used in this paper, and provides primitive conditions for the interval regression model.
Section \ref{loc_power_sec} derives the power results.
Section \ref{monte_carlo_sec} reports the results of a Monte Carlo study.
Section \ref{conclusion_sec} concludes.
An appendix contains
minimax power comparisons
as well as primitive conditions for the results in the main text in additional settings.
A supplementary appendix contains proofs and auxiliary results.

\subsection{Definition of Test Statistics}\label{test_def_sec}

The test statistics considered in this paper are as follows.  Given a set
$\mathcal{G}$ of nonnegative instruments, the null hypothesis %
$H_{0,\theta}:\theta\in\Theta_0$
implies that $E(m(W_i,\theta)g(X_i))\ge 0$ for all $g\in\mathcal{G}$.  Thus, under $H_{0,\theta}:\theta\in\Theta_0$,
the sample analogue
\begin{align}\label{inst_eq}
E_n(m(W_i,\theta)g(X_i))
  \equiv \frac{1}{n}\sum_{i=1}^n m(W_i,\theta)g(X_i)
\end{align}
should not be too negative for any $g\in\mathcal{G}$.  The results in this paper use classes of functions given by kernels with varying bandwidths and location, given by $\mathcal{G}=\{x\mapsto k((x-\tilde x)/h)|\tilde x\in\mathbb{R}^{d_X}, h\in\mathbb{R}_+\}$ for some kernel function $k$.
With this choice of $\mathcal{G}$,
$H_{0,\theta}:\theta\in\Theta_0$
holds if and only if $E(m(W_i,\theta)g(X_i))\ge 0$ for all $g\in\mathcal{G}$, so that (\ref{inst_eq}) can be used to form a consistent test \citep[see][for a discussion of this and other choices of $\mathcal{G}$]{andrews_inference_2013}.

Alternatively, one can test %
$H_{0,\theta}:\theta\in\Theta_0$
by estimating $E(m(W_i,\theta)|X_i=x)$ directly using the kernel estimate
\begin{align}\label{kern_eq}
\hat{\bar m}_j(\theta,x)
  =\frac{\sum_{i=1}^n m(W_i,\theta)k((X_i-x)/h)}{\sum_{i=1}^n k((X_i-x)/h)}
\end{align}
for some sequence $h=h_n\to 0$ and kernel function $k$.  If $H_{0,\theta}$ holds, (\ref{kern_eq}) should not be too negative for any $x$.

Thus, a test statistic of the null that $\theta\in\Theta_0$ can be formed by taking any function that is positive and large in magnitude when (\ref{inst_eq}) is negative and large in magnitude for some $g\in\mathcal{G}$, or when (\ref{kern_eq}) is negative and large in magnitude for some $x$.  One possibility is to use a CvM statistic that integrates the negative part of (\ref{inst_eq}) over some measure $\mu$ on $\mathcal{G}$.  This CvM statistic is given by
\begin{align}\label{iv_stat_eq}
T_{n,p,\omega,\mu}(\theta)
  =\left[\int \sum_{j=1}^{d_Y}
  |E_nm_j(W_i,\theta)g(X_i)\omega_j(\theta,g)|_{-}^p \,
    d\mu(g)\right]^{1/p}
\end{align}
for some $p\ge 1$ and weighting $\omega$,
where $|t|_{-}=|\min\{t,0\}|$.
I refer to this as an instrument based CvM (IV-CvM) statistic.
The CvM statistic based on the kernel estimate integrates the negative part of (\ref{kern_eq}) against some weighting $\omega$, and is given by
\begin{align}\label{kern_stat_eq}
T_{n,p,\text{kern}}(\theta)
=\left[\int \sum_{j=1}^{d_Y} \left| \hat{\bar m}_j(\theta,x)\omega_j(\theta,x)
   \right|_{-}^p \, d x\right]^{1/p}
\end{align}
for some $p\ge 1$.
I refer to this as a kernel based CvM (kern-CvM) statistic.

For the instrument based CvM statistic, the scaling for the power function will depend on $\omega$.  This paper considers both
a bounded weighting which, without loss of generality, can be taken to be constant (the measure $\mu$ can absorb any weighting that does not change with the sample size)
\begin{align}\label{bdd_weight_eq}
\omega_j(\theta,g)=1 \text{ all $\theta,g,j$}
\end{align}
as well as the truncated variance weighting 
used for KS statistics by \citet{armstrong_weighted_2014}, \citet{armstrong_multiscale_2016} and
\citet{chetverikov_adaptive_2012},
which is given by
\begin{align}\label{var_weight_eq}
\omega_j(\theta,g)=(\hat \sigma_j(\theta,g)\vee \sigma_n)^{-1}
\end{align}
where
\begin{align*}
\hat\sigma_j(\theta,g)
  =\{E_n[m_j(W_i,\theta)g(X_i)]^2-[E_nm_j(W_i,\theta)g(X_i)]^2\}^{1/2}
\end{align*}
and $\sigma_n$ is a sequence converging to zero and $a\vee b$ denotes the maximum of $a$ and $b$ for scalars $a$ and $b$.\footnote{For the critical value of the test, the results covered in this paper cover any critical value that is of the same order of magnitude asymptotically as a critical value based on the distribution where all moments bind.  See Section \ref{def_sec} for details.}

The results for CvM statistics derived in this paper can be compared to power results for KS statistics derived in \citet{armstrong_asymptotically_2015} and \citet{armstrong_weighted_2014}.  A KS statistic based on (\ref{inst_eq}) simply takes the most negative value of that expression over $g\in\mathcal{G}$, and is given by
\begin{align}\label{iv_stat_ks_eq}
T_{n,\infty,\omega}(\theta)
  =\max_j \sup_{g\in\mathcal{G}} 
  |E_nm_j(W_i,\theta)g(X_i)\omega_j(\theta,g)|_{-}.
\end{align}
I refer to this as an instrument based KS (IV-KS) statistic.
A KS statistic based on (\ref{kern_eq}) simply takes the most negative value of that expression over $x$, and is given by
\begin{align}\label{kern_stat_ks_eq}
T_{n,\infty,\text{kern}}(\theta)
=\max_j\sup_{\theta} \left| \hat{\bar m}_j(\theta,x)\omega_j(\theta,x)\right|_{-}.
\end{align}
I refer to this as a kernel based KS (kern-KS) statistic.
As with CvM statistics, the scaling for the local power function for the instrument based KS test depends on whether a bounded weighting or a truncated variance weighting is used.

To complete the definition of these tests, we need to define a critical value.
For tests that use instrument based CvM statistics with bounded weights or inverse variance weights with $p<\infty$, the test $\phi_{n,p,\omega,\mu}(\theta)$, which rejects when $\phi_{n,p,\omega,\mu}(\theta)=1$, is defined as
\begin{align}\label{test_def_rootn}
\phi_{n,p,\omega,\mu}(\theta)
=\left\{\begin{array}{cc}
1 & \text{if }\sqrt{n}T_{n,p,\omega,\mu}(\theta) > \hat c_{n,p,\omega,\mu}(\theta)  \\
0 & \text{otherwise}
\end{array}\right.
\end{align}
for some critical value $\hat c_{n,p,\omega,\mu}(\theta)$.  For kernel based CvM statistics, the test $\phi_{n,p,\text{kern}}(\theta)$, which rejects when $\phi_{n,p,\text{kern}}(\theta)=1$, is defined as
\begin{align}\label{test_def_kern}
\phi_{n,p,\text{kern}}(\theta)
=\left\{\begin{array}{cc}
1 & \text{if }(nh^{d_X})^{1/2} T_{n,p,\text{kern}}(\theta) > \hat c_{n,p,\text{kern}}(\theta)  \\
0 & \text{otherwise}
\end{array}\right.
\end{align}
While all of the new results in this paper are for CvM statistics, I refer to analogous results for KS statistics at some points for comparison.  For KS tests with bounded weights, the critical value is defined as in (\ref{test_def_rootn}).
For KS tests based on truncated variance weights, the test
$\phi_{n,\infty,(\sigma\vee \sigma_n)^{-1}}(\theta)$ is defined as
\begin{align}\label{test_def_rootlogn}
\phi_{n,\infty,(\sigma\vee \sigma_n)^{-1}}(\theta)
=\left\{\begin{array}{cc}
1 & \text{if }\sqrt{\frac{n}{\log n}}T_{n,\infty,(\sigma\vee \sigma_n)^{-1}}(\theta) > \hat c_{n,\infty,(\sigma\vee \sigma_n)^{-1}}(\theta)  \\
0 & \text{otherwise}
\end{array}\right.
\end{align}
for some critical value $\hat c_{n,p,\infty,(\sigma\vee \sigma_n)^{-1}}(\theta)$.

\subsection{Related Literature}\label{related_lit_sec}

Tests based on instrument based CvM and KS statistics have been considered by \citet{andrews_inference_2013}, \citet{kim_kyoo_il_set_2008},
\citet{khan_inference_2009} and \citet{armstrong_asymptotically_2015}
for bounded weights, and \citet{armstrong_weighted_2014},
\citet{armstrong_multiscale_2016}
and \citet{chetverikov_adaptive_2012}
for KS statistics with variance weights.
The statistics based on instruments with bounded weights use an approach to nonparametric testing problems that goes back at least to \citet{bierens_consistent_1982}.
\citet{aradillas-lopez_testing_2013} use a slightly different version of an instrument CvM approach.
\citet{chernozhukov_intersection_2013} consider kernel based KS statistics and
\citet{lee_testing_2013}
and
\citet{lee_testing_2015}
consider kernel based CvM statistics.  While some of these papers
derive local power results for CvM tests under conditions that appear
to be common in point identified models, these results do not apply in
set identified models except for in very special cases.  Indeed, the results in the present paper show that, when one uses a minimax criterion requiring uniformly good power in classes of underlying distributions defined by smoothness properties, the power of CvM tests is much worse (see Section \ref{minimax_rates_sec}).  The results
in this paper show that power comparisons in the set identified case
considered here are much different than settings that have been studied
previously.  \citet{armstrong_asymptotically_2015},
\citet{armstrong_weighted_2011},
\citet{armstrong_weighted_2014},
\citet{armstrong_multiscale_2016},
 and \citet{chetverikov_adaptive_2012}
derive power results for KS statistics under conditions similar to
those used in this paper, but do not consider CvM statistics.

The results in this paper are also related to the statistics literature on minimax testing of hypotheses of the form $H_{0,=}: f(x)=0$ all $x$, $H_{0,\ge}: f(x)\ge 0$ all $x$, $H_{0,\uparrow}: f(x)\ge f(x')$ all $x<x'$, (and related hypotheses such as convexity of $f$), where the function $f$ is observed with noise.  While much of this literature focuses on the Gaussian white noise model or Gaussian sequence model, the results are closely related to the case where $f(x)=E(Y_i|X_i=x)$, and iid observations of $X_i,Y_i$ are available (which falls into our setup if we take $Y_i=m(W_i,\theta_0)$).
To formulate the minimax testing problem considered in this literature, one specifies a smoothness class $\mathcal{F}$ for $f$ and a functional $\psi:\mathcal{F}\to [0,\infty)$ such that $\psi(f)$ is $0$ if $f$ satisfies the null and strictly positive otherwise.
For example, for $H_{0,=}$, one can take the $L_p$ norm $\psi(f)=[\int f(x)^p\, dx]^{1/p}$ and, for $H_{0,\ge}$, one can take the one-sided $L_p$ norm $\psi(f)=[\int |f(x)|_{-}^p]$.
The minimax testing problem is to obtain tests that have good worst-case power over alternatives $f$ in the smoothness class $\mathcal{F}$ with $\psi(f)\ge a_n$ for $a_n\to 0$ as quickly as possible.  \citet{dumbgen_multiscale_2001} and \citet{juditsky_nonparametric_2002} consider $H_{0,\ge}$ with $\psi$ given by the one-sided $L_\infty$ norm $\psi(f)=\sup_x |f(x)|_{-}$ and the one-sided $L_p$ norm with $p<\infty$ respectively, as well as $H_{0,\uparrow}$ and the hypothesis of convexity with related distance functions $\psi$.  \citet{lepski_asymptotically_2000} consider $H_{0,=}$ with $\psi(f)$ given by the $L_\infty$ norm and by $\psi(f)=|f(x_0)|$ for a given point $x_0$.  See \citet{ingster_nonparametric_2003} for further results and references to this literature.

In contrast to this literature, the results in this paper have implications for minimax rates of CvM statistics for testing the null that a given value of $\theta$ is in the identified set against the alternative that the distance between $\theta$ and any point in the identified set is at least $a_n$ (see Section \ref{minimax_rates_sec} in the appendix for a formal statement).  Since the dimension of $\theta$ is finite and fixed, the choice of distance (i.e. whether to use Euclidean distance or sup-norm distance when defining distance of $\theta$ from points in the identified set) does not matter for the rate at which $a_n$ can approach zero with the test having good power.  This contrasts with the nonparametric testing literature described above, in which the choice of distance function $\psi$ has implications for relative efficiency of different test statistics, and is part of the reason that CvM and KS tests can be ranked in this setting.  Interestingly, the problem of minimax inference on $\theta$ in the settings considered here appears to be closely related to nonparametric testing with $\psi$ given by the $L_\infty$ norm.  See \citet{armstrong_note_2014} for further discussion.

\section{Intuition for the Results}\label{intuition_sec}

To get some intuition for the results, consider the interval regression model defined in the introduction, with a one-dimesional covariate $X_i$.  Consider a sequence $\theta_n$ converging to a parameter $\theta_0=(\theta_{0,1},\theta_{0,2})'$ on the boundary of the identified set such that $(1,x)\theta_0=\theta_{0,1}+\theta_{0,2}x$ is tangent to $E(W_i^H|X_i=x)$ at some point $x_0$.  This is illustrated in Figure \ref{local_alt_fig} (the conditional mean of $W_i^L$ can be considered to be below the range depicted in the figure).
To keep the derivations below simple, we assume that $\theta_n$ is formed by adding a sequence $a_{n,1}$ to the constant term in $\theta_0$, so that $\theta_n=(\theta_{0,1}+a_{n,1},\theta_{0,2})'$ where $\theta_0=(\theta_{0,1},\theta_{0,2})'$.  However, our general results cover parameter sequences where the intercept changes with $n$ as well.

The test statistics $T_n(\theta_n)$ defined in Section \ref{test_def_sec} will take sample analogues of $E((W_i^H-(1,x)\theta_n)g(X_i))$ for functions $g$ of the form $g(X_i)=k((X_i-\tilde x)/h)$ for some $\tilde x$ and $h$, and integrate or take the minimum of those that are negative.
In order for the test $\phi_n(\theta_n)$ based on a test statistic $T_n(\theta_n)$ to have high power, we would like the test to place as much weight as possible on functions $g(X_i)=k((X_i-\tilde x)/h)$ that are supported on values of $X_i$ where the inequality is violated (in the case of the parameter $\theta_n$ illustrated in the figure, this corresponds to $X_i$ between about $.5$ and $.7$).
As $\theta_n$ approaches $\theta_0$, the portion of the support of $X_i$ where the inequality is violated will shrink towards a single point $x_0$, so we will want the test statistic to use functions $g(X_i)=k((X_i-\tilde x)/h)$ where $\tilde x$ is near $x_0$ and $h$ is close to zero.

If we knew a priori the point where the moment inequality was violated, we could use this to choose the function $g(X_i)$.  The fact that this is unknown leads to the tests described in Section \ref{test_def_sec}, where test statistics are formed by combining these functions using integration (for CvM statistics) or by taking the maximum (for KS statistics).  This is the step that leads to CvM and KS statistics having different power properties: taking the integral of the moment functions tends to give power when the inequality is violated by a small amount at many different points, while taking the maximum leads to more power when the inequality is violated at a small number of points.  Since the moment inequality is violated on a shrinking set, KS statistics have better power in this setting than CvM statistics.\footnote{To see this in a simpler setting, consider testing a finite number of unconditional moment inequalities $H_0:EY_{i,1}\ge 0,\ldots,E Y_{i,k}\ge 0$.  Tests based on the statistic $\sum_{j=1}^k \left|\sum_{i=1}^n Y_{i,k}\right|_-^2$ (which is analogous to a CvM statistic) will have more power when each of the inequalities is violated by a small amount, while tests based on the statistic $\max_{j=1}^k \left|\sum_{i=1}^n Y_{i,j}\right|_-$ (which is analogous to a KS statistic) will have more power when a single inequality is violated.  See \citet{armstrong_note_2014} for details and further references.}

To see this in more detail, let us give a heuristic derivation of some of the results in this setting.  Consider the instrument based CvM statistic with bounded weights, where the measure $\mu$ on the instruments $g(x)=k((x-\tilde x)/h)$ has a density $f_\mu(\tilde x,h)$ with respect to the Lebesgue measure, and assume that $X_i$ has a density $f_X(x)$.  For simplicity, suppose we only base the statistic on the inequality involving $W_i^H$.  The statistic is $T_n(\theta_n)=\left[\int\int|E_n(W_i^H-(1,X_i)\theta_n)k((X_i-\tilde x)/h)|_-^pf_\mu(\tilde x,h)\, d\tilde x\, dh\right]^{1/p}$,
which is an integral over a sample expectation.
We expect that the test will have power when the
integral over the corresponding population expectation is large
relative to the critical value, which, as discussed below, will be of
order $n^{-1/2}$.  Thus, to have power at $\theta_n=(\theta_{0,1}+a_{n,1},\theta_{0,2})'$, we
expect that
\begin{align}\label{drift_eq_intuition}
&\left[\int \int 
  |E(W_i^H-(1,X_i)\theta_n)k((X_i-\tilde x)/h)|_{-}^p
f_\mu(\tilde x,h)\, d\tilde x \, dh
\right]^{1/p}  \nonumber  \\
&=\left[\int \int 
  \left|\int (E(W_i^H|X_i=x)-(1,x)\theta_n)k((x-\tilde x)/h) f_X(x)\,dx\right|_{-}^p
f_\mu(\tilde x,h)\, d\tilde x \, dh
\right]^{1/p}
\end{align}
will have to be large relative to $n^{-1/2}$.

Since $E(W_i^H|X_i=x)$ is tangent to $(1,x)\theta_0=\theta_{0,1}+\theta_{0,2}x$ at $x_0$, a second order Taylor approximation gives $E(W_i^H|X_i=x)-\theta_{0,1}-\theta_{0,2}x\approx (x-x_0)^2 (V/2)$ where $V$ is the second derivative of $E(W_i^H|X_i=x)$ at $x_0$.  Since $\theta_n=(\theta_{0,1}+a_{n,1},\theta_{0,2})'$, this gives an approximation to the integrand in the above display:
$E(W_i^H|X_i=x)-(1,x)\theta_n=E(W_i^H|X_i=x)-\theta_{0,1}-\theta_{0,2}x - a_{n,1}\approx (x-x_0)^2 (V/2) - a_{n,1}$.  Substituting this into the above display gives
\begin{align*}
\left[\int \int 
  \left|\int ((x-x_0)^2(V/2)-a_{n,1})k((x-\tilde x)/h) f_X(x)\,dx\right|_{-}^p
f_\mu(\tilde x,h)\, d\tilde x \, dh
\right]^{1/p}.
\end{align*}
As $\theta_n$ approaches $\theta_0$, only values of $\tilde x$ near $x_0$ and values of $h$ near zero will contribute to the integrand, so that this approximation will hold with increasing accuracy.  Furthermore, assuming that $f_\mu$ and $f_X$ are smooth, this means that we can also replace $f_X(x)$ with $f_X(x_0)$ and $f_\mu(\tilde x,h)$ with $f_\mu(x_0,0)$:
\begin{align*}
\left[\int \int 
  \left|\int ((x-x_0)^2(V/2)-a_{n,1})k((x-\tilde x)/h) f_X(x_0)\,dx\right|_{-}^p
f_\mu(x_0,0)\, d\tilde x \, dh
\right]^{1/p}.
\end{align*}
Using the
change of variables
$u=(x-x_0)/a_{n,1}^{1/2}$,
$v=(\tilde x-x_0)/a_{n,1}^{1/2}$, $\tilde h=h/a_{n,1}^{1/2}$,
it can be seen that the above display is equal to
\begin{align*}
&\left[\int \int 
  \left|\int (a_{n,1}u^2(V/2)-a_{n,1})k((u-v)/\tilde h) f_X(x_0)a_{n,1}^{1/2}\,du\right|_{-}^p
f_\mu(x_0,0)a_{n,1}^{1/2} \,d\tilde v a_{n,1}^{1/2}\, d\tilde h
\right]^{1/p}  \\
&=a_{n,1}^{3/2+1/p}\left[\int \int 
  \left|\int (u^2(V/2)-1)k((u-v)/\tilde h) f_X(x_0)\,du\right|_{-}^p
f_\mu(x_0,0) \,d\tilde v \, d\tilde h
\right]^{1/p}.
\end{align*}
Thus, we expect to get power when $a_{n,1}^{3/2+1/p}$ decreases at least as slowly as $n^{-1/2}$, which corresponds to $a_{n,1}$ decreasing at the rate $n^{-1/(3+2/p)}$.  This is the rate derived formally for this test in Section \ref{inst_bdd_sec}, specialized to this setting (the general results use a smoothness parameter $\gamma$ which, in this case, is equal to $2$).

To understand how this differs from the corresponding KS test based on $T_n(\theta_n)=\sup_{\tilde x,h}|E_n(W_i^H-(1,X_i)\theta_n)k((X_i-\tilde x)/h)|_-$, note that similar derivations give the approximation
\begin{align*}
\sup_{\tilde x,h}
  \left|\int ((x-x_0)^2(V/2)-a_{n,1})k((x-\tilde x)/h) f_X(x_0)\,dx\right|_{-}.
\end{align*}
Applying the same change of variables gives
\begin{align*}
&\sup_{u,\tilde h}\left|\int (a_{n,1}u^2(V/2)-a_{n,1})k((u-v)/\tilde h) f_X(x_0)a_{n,1}^{1/2}\,du\right|_{-}  \\
&=a_{n,1}^{3/2}\sup_{u,\tilde h}\left|\int (u^2(V/2)-1)k((u-v)/\tilde h) f_X(x_0)\,du\right|_{-},
\end{align*}
and comparing this to $n^{-1/2}$ (which is the order of the critical value in this case as well) shows that we will have power when $a_{n,1}$ decreases at the rate $n^{-1/3}$.  This is shown formally in \citet{armstrong_asymptotically_2015}.  Note that the $n^{-1/3}$ rate for the KS statistic is faster than the $n^{-1/(3+2/p)}$ rate for the CvM statistic.

\section{Assumptions}\label{larger_assumptions_sec}

This section states the conditions used in this paper, and verifies them for the interval regression model defined in the introduction.  Section \ref{prim_cond_append} in the appendix verifies the conditions in other settings.

This paper considers the power $P(\phi_n(\theta_n)=1)$ of a sequence $\phi_n(\theta_n)$ of tests of $H_{0,\theta_n}:\theta_n\in\Theta_0$ under iid data from a fixed dgp $P$, where $\theta_n=\theta_0+a_n$ is a sequence converging to $\theta_0$ on the boundary of $\Theta_0$ (where $\Theta_0$ is the identified set under the given dgp $P$).
Thus, we need conditions on the tests $\phi_n(\theta_n)$ (in particular, the critical values and weighting functions, etc. used in forming the test statistics) and the dgp $P$ and the sequence $\theta_n$.  Section \ref{def_sec} gives the conditions on the tests $\phi_n(\theta_n)$ and Section \ref{assumptions_sec} gives the conditions on $P$ and $\theta_n$.  Section \ref{int_reg_conditions_sec} verifies these conditions for the interval regression model.  Section \ref{rootn_comparison_sec} explains how the conditions differ from those encountered in point identified settings.

\subsection{Assumptions on Test Statistics and Critical Values}\label{def_sec}

The properties of these tests will depend on the choice of critical value.
The only condition needed for upper bounds on power, stated in the following assumption, is that the critical value be of the same order of magnitude as a critical value based on a least favorable asymptotic distribution where all of the moments bind (i.e. $E(m(W_i,\theta)|X_i)=0$ a.s.).

\begin{assumption}\label{cval_bound_assump}
For some $\eta>0$,
the critical value $\hat c_n=\hat c_n(\theta_n)$ defined in (\ref{test_def_rootn}) or (\ref{test_def_kern}), depending on the weighting and form of the test, satisfies $\hat c_n(\theta_n)>\eta$ with probability approaching one.
\end{assumption}

Assumption \ref{cval_bound_assump} holds for the kernel CvM based test of \citet{lee_testing_2013}, which uses the least favorable null dgp, as well as the tests using instrument based CvM statistics with bounded weights proposed in \citet{andrews_inference_2013}.
Instrument based CvM statistics with variance weights have not been considered in the literature.  In Section \ref{cval_append} of the supplementary appendix, I consider critical values for this case and show that critical values based on the least favorable null dgp will satisfy Assumption \ref{cval_bound_assump}.

Assumption \ref{cval_bound_assump} only gives a lower bound for a critical value.  This gives bounds on the power, but to derive the exact local asymptotic power, we need the following condition, which gives a limiting value for this critical value.  Under mild conditions on the data generating process and sequence of local alternatives, this assumption will also hold for the methods of choosing critical values discussed above.

\begin{assumption}\label{cval_limit_assump}
For the critical value $\hat c_n=\hat c_n(\theta_n)$ defined in (\ref{test_def_rootn}) or (\ref{test_def_rootlogn}), depending on the weighting and form of the test, and some constant $c>0$,
$\hat c_n(\theta_n)\stackrel{p}{\to} c$.
\end{assumption}

The power properties of the test will also depend on the class of
functions $\mathcal{G}$ used as instruments.  I derive power results for
the case where $\mathcal{G}$ consists of kernel functions with different
bandwidths and locations, defined in the following assumption.

\begin{assumption}\label{g_kernel_assump}
For some bounded, nonnegative function $k$ 
with finite support and $\int k(u)\, du>0$,
$\mathcal{G}=\{x\mapsto k((x-\tilde x)/h)|\tilde
x\in\mathbb{R}^{d_X}, h\in\mathbb{R}_+\}$, and
the covering number $N(\varepsilon,\mathcal{G},L_1(Q))$ defined in
\citet{pollard_convergence_1984} satisfies
$\sup_{Q} N(\varepsilon,\mathcal{G},L_1(Q)) \le A\varepsilon^{-W}$,
where the supremum is over all
probability measures.
\end{assumption}

The covering number assumption in Assumption \ref{g_kernel_assump} is a technical condition that allows for uniform convergence of kernel estimates over $x$ and $h$.  A sufficient condition is that the kernel $k$ takes the form $k(x)=r(\|x\|)$ where $r$ is a monotone decreasing function on on $[0,\infty)$ \citep[see][chapter 2, problem 28]{pollard_convergence_1984}.

For CvM statistics, I place the following condition on the measure $\mu$ over which the sample means are integrated.

\begin{assumption}\label{mu_assump}
The measure $\mu$ has bounded support, and has a density $f_\mu(\tilde x, h)$ with respect to the Lebesgue measure on $\mathbb{R}^{d_X}\times [0,\infty)$ that is bounded and continuous.
\end{assumption}

Relaxing this assumption would lead to different power properties, although the general point that $L_p$ statistics perform worse in these models than supremum statistics would go through.

\subsection{Conditions on Data Generating Process}\label{assumptions_sec}

This section presents the main assumptions on the model and dgp used in this paper.
The conditions are similar to those used in
\citet{armstrong_asymptotically_2015},
\citet{armstrong_weighted_2014} and
\citet{armstrong_multiscale_2016}.
I first provide high level conditions, and then verify them for the interval regression model in Section \ref{int_reg_conditions_sec}.  Section \ref{rootn_comparison_sec} provides a discussion of the difference between these conditions and other settings, such as point identified models.
Section \ref{prim_cond_append} in the appendix verifies the conditions in this section for additional settings.
I assume throughout that the data are iid.

I place the following conditions on the data generating process
and the sequence $\theta_n=\theta_0+a_n$.
In these conditions, $\gamma$ is a smoothness parameter that is generally given by the minimum of the number of derivatives of the conditional mean and $2$.  The truncation of the smoothness parameter at $2$ comes from the fact that the test statistics here use positive kernels or instruments.

\begin{assumption}\label{smoothness_assump_multi}
For each $j$, the conditional mean $E(m_j(W_i,\theta_0)|X_i=x)\equiv \bar m_j(\theta_0,x)$
takes its minimum only on a finite set
$\{x|E(m_j(W_i,\theta_0)|X_i=x)=0 \text{ some
  $j$}\}=\mathcal{X}_0=\{x_1,\ldots,x_\ell\}$.  For each $k$ from $1$
to $\ell$, let $J(k)$ be the set of indices $j$ for which
$E(m_j(W_i,\theta_0)|X_i=x_k)=0$.
Assume that there exist neighborhoods $B(x_k)$ of each
$x_k\in\mathcal{X}_0$
such that
the following assumptions hold.
\begin{itemize}
\item[i.)] There exists $\eta>0$ such that, for $\theta$ in a neighborhood of $\theta_0$, we have (a) $\bar m_j(\theta,x)>\eta$ for $j\notin J(k)$ for $x\in B(x_k)$
and (b) $\bar m_j(\theta,x)>\eta$ for all $j$ for $x\notin \cup_{k=1}^\ell B(x_k)$.

\item[ii.)]
For $j\in J(k)$,
$\bar m_j(\theta_0,x)$ is continuous
on
the closure of
$B(x_k)$ and satisfies
\begin{align*}
\sup_{\|x-x_k\|\le \delta} \left\|
  \frac{\bar m_j(\theta_0,x)-\bar m_j(\theta_0,x_k)}{\|x-x_k\|^{\gamma(j,k)}}
  -\psi_{j,k}\left(\frac{x-x_k}{\|x-x_k\|}\right) \right\|
\stackrel{\delta \to 0}{\to} 0
\end{align*}
for some $\gamma(j,k)>0$ and some function
$\psi_{j,k}:\{t\in\mathbb{R}^{d_X}|\|t\|=1\}\to \mathbb{R}$ with
$\overline \psi \ge \psi_{j,k}(t)\ge \underline \psi$ for some
$\overline \psi <\infty$ and $\underline \psi>0$.
For future reference, define
$\gamma=\max_{j,k}\gamma(j,k)$ and $\tilde J(k)=\{j\in J(k)|
\gamma(j,k)=\gamma\}$.

\item[iii.)] $X_i$ has a continuous density $f_X$ on $B(x_k)$.

\item[iv.)] For %
$j\in J(k)$, $s_j^2(x,\theta)\equiv var(m_j(W_i,\theta)|X_i=x)$ is strictly positive and continuous at $(x_k,\theta_0)$.

\item[v.)]
For $x$ in the closure of $B(x_k)$ and $\theta$ in a neighborhood of $\theta_0$,
$\bar m(\theta,x)$ has a derivative as a function of
$\theta$
that is continuous as a
function of $(\theta,x)$.  Let $\bar m_{\theta,j}(\theta,x)$ denote the $j$th row of this derivative matrix (i.e. the derivative of $\bar m_j(\theta,x)$ with respect to $\theta$).
\end{itemize}
\end{assumption}

\begin{assumption}\label{bdd_y_assump_local}
The data are iid and for some fixed $\overline Y<\infty$ and $\theta$ in a
some neighborhood of
$\theta_0$, $|m(W_i,\theta)|\le \overline Y$ with probability one.
\end{assumption}

The deterministic bound in Assumption \ref{bdd_y_assump_local} allows for the use of certain technical results that are useful in the proofs.  It may be possible to relax this assumption, although additional technical arguments would be needed in some places.

The following assumption, which is used for kernel based statistics,
ensures that the kernel estimators do not encounter boundary problems \citep[cf. Assumption 1(iii) in][]{lee_testing_2013}.
\begin{assumption}\label{kern_dens_assump}
$X_i$ has a density $f_X$ that is bounded away from
infinity,
and the weighting function $\omega_j(\theta,x)$ is continuous for all $j$ and, for some $\varepsilon>0$, is equal to zero whenever $f_X(\tilde x)<\varepsilon$ for some $\tilde x$ with $\|\tilde x-x\|<\varepsilon$.
\end{assumption}

\subsection{Discussion and Primitive Conditions for Interval Regression}\label{int_reg_conditions_sec}

In discussing these assumptions, it is useful to keep in mind the interval regression model introduced in the introduction, in which $W_i=(X_i,W_i^L,W_i^H)$ and $m(W_i,\theta)=(W_i^H-(1,X_i')\theta,(1,X_i')\theta-W_i^L)'$.  The following gives a general discussion of these assumptions, with references to the interval regression model as an example.  I then state primitive sufficient conditions in the interval regression model that imply these assumptions with $\gamma=2$.  Section \ref{prim_cond_append} of the appendix gives primitive conditions in additional settings.

The assumptions used here are similar to
the conditions used in \citet{armstrong_asymptotically_2015} to derive the asymptotic distribution and local power of a KS statistic with bounded weights.  In particular, Assumption \ref{smoothness_assump_multi} corresponds to
the version of Assumption 3.1 in \citet{armstrong_asymptotically_2015} used in Section 5 of that paper, in which part (ii) is replaced by Assumption 5.1 in \citet{armstrong_asymptotically_2015}.
Part (i) strengthens the version used in \citet{armstrong_asymptotically_2015} by extending it to a neighborhood of $\theta_0$, and part (v) is an additional condition on the derivative with respect to $\theta$.  These additional conditions are used to derive local power, and are similar to Assumption 7.1 in \citet{armstrong_asymptotically_2015}.

Assumption \ref{smoothness_assump_multi} is the main substantive condition that gives rise to the local power results derived in this paper.
It states that the conditional mean of the moment conditions is equal to zero only at a finite number of points.  In the context of the interval regression model, this holds for $\theta_0$ on the boundary of the identified set when the regression line $x'\theta_0$ is tangent to $E(W_i^H|X_i=x)$ or $E(W_i^H|X_i=x)$ at a finite number of points.  In general, a sufficient condition for this in the case where $X_i$ has compact support is that $\bar m_j(\theta,x)$ takes its minimum on the interior of the support of $X_i$ and $\bar m_j(\theta,x)$ is twice continuously differentiable with a positive definite second derivative matrix at any point where it takes a minimum (see Section \ref{finite_contact_set_sec} in the appendix).

The most natural case where this does not hold is where $E(W_i^H|X_i=x)$ or $E(W_i^L|X_i=x)$ is linear and equal to $(1,x')\theta$ on a nondegenerate interval (the other possibility is for $E(W_i^H|X_i=x)-(1,x')\theta_0$ to be zero on a set with infinitely many elements, but with zero probability, such as with the function $\sin (1/x)$).  This holds in the point identified case where $P(W_i^H=W_i^L|X_i)=1$ for $X_i$ on a nondegenerate interval (and, in particular, in the special case where $W_i^H=W_i^L$ with probability one, leading to the usual linear regression model).  However, when $\theta$ is set identified, this is a knife-edge case: even if $E(W_i^H|X_i)=(1,X_i')\theta_0$ for $X_i$ on a nondegenerate interval for a given $\theta_0$ on the boundary of the identified set, we will typically have $E(W_i^H|X_i=x)=(1,x')\tilde\theta_0$ only on a finite set for $\tilde\theta_0$ close to $\theta_0$.

This is illustrated by Figures \ref{int_reg_smooth_fig} and \ref{int_reg_rootn_fig}, which are taken directly from Section 2.2 of \citet{armstrong_asymptotically_2015}.
Each figure shows the conditional mean $E(W_i^H|X_i=x)$ for some dgp along with regression lines corresponding to particular parameter values $\theta$ (the lower conditional mean $E(W_i^L|X_i=x)$ can be taken to be below the area shown in each figure).  In Figure \ref{int_reg_smooth_fig}, the regression line $(1,x')\theta=\theta_1+\theta_2x$ is tangent to the conditional mean at a single point, and Assumption \ref{smoothness_assump_multi} holds for the parameter $\theta$.  In Figure \ref{int_reg_rootn_fig}, the regression line $\theta_{a,1}+\theta_{a,2}x$ corresponding to the parameter $\theta_a$ is equal to $E(W_i^H|X_i=x)$ on a nondegenerate interval, so that Assumption \ref{smoothness_assump_multi} does not hold.  However, at nearby parameter values such as $\theta_b$, the regression line is equal to $E(W_i^H|X_i=x)$ at a single point and Assumption \ref{smoothness_assump_multi} holds.
See Section 2.2 of \citet{armstrong_asymptotically_2015} for further discussion.

In the case where $\bar m(\theta_0,x)$ is twice continuously differentiable in $x$, part (ii) of Assumption \ref{smoothness_assump_multi} follows from a second order Taylor expansion at $x_k$, so long as the second derivative matrix is positive definite.  In this case, Assumption \ref{smoothness_assump_multi} holds with $\gamma=2$ and $\psi_{j,k}(u)=u'V_j(x_k)u/2$, where $V_j(x_k)$ is the second derivative matrix of $x\mapsto \bar m(\theta_0,x)$ at $x_k$.  In the interval regression model, the second derivative of $m_1(\theta_0,x)$ is equal to the second derivative of $E(W_i^H|X_i=x)$ (and similarly for $m_2(\theta_0,x)$ and $-E(W_i^L|X_i=x)$), so this translates directly to an assumption of a positive definite second derivative matrix of $E(W_i^H|X_i=x)$.
In the case where $\bar m(\theta_0,x)$ is Lipschitz continuous, part (ii) of Assumption \ref{smoothness_assump_multi} will hold with $\gamma=1$ if we place additional regularity conditions on the one-sided directional derivative of $\bar m(\theta_0,x)$. 
The parameter $\theta$ in Figure \ref{int_reg_smooth_fig} illustrates a case where Assumption \ref{smoothness_assump_multi} holds with $\gamma=2$, while the parameter $\theta_b$ in Figure \ref{int_reg_rootn_fig} illustrates a case where Assumption \ref{smoothness_assump_multi} holds with $\gamma=1$.
See Theorem \ref{int_reg_lipschitz_thm} in Section \ref{int_reg_sec} of the appendix for a formal statement in the interval regression model.

The remaining assumptions are regularity conditions that translate easily to primitive objects in the case of interval regression.  For part (v), note that  $\bar m_{\theta,1}(\theta,x)=-(1,x')$ and $\bar m_{\theta,2}(\theta,x)=(1,x')$, which are clearly continuous, so this assumption holds without further conditions on the dgp.

The following gives a formal statement of primitive conditions for the interval regression model in the case where the conditional means are twice differentiable.  The proof of this result uses the ideas in the discussion above, and is given in Section \ref{int_reg_sec} of the appendix.

\begin{theorem}\label{int_reg_second_deriv_thm}
Suppose that the following conditions hold.
\begin{itemize}
\item[i.)] The conditional means $E(W_i^H|X_i=x)$ and $E(W_i^L|X_i=x)$ are twice differentiable with continuous second derivatives, $X_i$ has a continuous density and compact support, and $W_i^H$ and $W_i^L$ are bounded from above and below by finite constants.

\item[ii.)]
For any point $\tilde x$ such that $E(W_i^H|X_i=\tilde x)=(1,\tilde x')\theta_0$,
$\tilde x$ is in the interior of the support of $X_i$, $var(W_i^H|X_i=x)$ is positive and continuous at $\tilde x$ and
$E(W_i^H|X_i=x)$ has a positive definite second derivative matrix at $\tilde x$.
The same holds for $E(W_i^L|X_i=x)$ with ``positive definite'' replaced by ``negative definite.''
\end{itemize}

Then Assumptions
\ref{smoothness_assump_multi},
and \ref{bdd_y_assump_local} hold, with $\gamma=2$ in Assumption \ref{smoothness_assump_multi}.
\end{theorem}

\subsection{Comparison with Conditions Leading to Parametric Rates}\label{rootn_comparison_sec}

Under Assumption \ref{smoothness_assump_multi}, the conditional mean $\bar m_j(\theta_0,x)=E(m(W_i,\theta_0)|X_i=x)$ is minimized on a finite set, and behaves like $\|x-x_k\|^\gamma$ for $x_k$ in this set and nearby $x$.  As shown in Section \ref{loc_power_sec} below, this leads to power against alternatives that approach the identified set at a slower than $\sqrt{n}$ rate.  As suggested by the intuitive description of these results in Section \ref{intuition_sec}, this arises because, as $\theta_n$ approaches the identified set, the conditional moment inequalities are violated on a set with vanishing probability.  This is similar to the case of nonparametric kernel estimation, in which bias-variance tradeoffs and the level of smoothness determine the rate of convergence \citep[see, e.g.,][]{wasserman_all_2007}.

In contrast, \citet{andrews_inference_2013}, \citet{kim_kyoo_il_set_2008} and \citet{lee_testing_2013} consider the case where $\bar m_j(\theta_0,x)$ is minimized on a nondegenerate interval.  In this case, the portion of the support of $X_i$ on which the inequality is violated does not vanish as $\theta_n$ approaches the boundary of the identified set.  This leads to nontrivial power at alternatives that approach the null at a $1/\sqrt{n}$ rate.
As discussed above, the latter case is typical under point identification and holds by construction with moment equalities, but it corresponds to a knife-edge case under set identification.

To understand these issues, it is helpful to make a comparison to the case of nonparametric regression, where kernel estimators can converge at a faster rate if certain derivatives are equal to zero.  For example, local linear estimators converge at a $n^{2/5}$ rate when the conditional mean is twice differentiable with nonzero derivative and a bandwidth is used that decreases like $n^{-1/5}$, but a faster rate can be obtained when the second derivative is zero, using a bandwidth sequence that converges more slowly.
The typical approach to formalizing the notion that the optimal rate under a second derivative condition is $n^{2/5}$ is to use a minimax criterion, in which one requires good performance uniformly over all dgps with a certain bound on the second derivative \citep[see][for a formulation of this approach for local linear estimators]{fan_local_1993}.
Minimax results of this form are often cited in econometrics when making claims of optimality of nonparametric estimators (for example, \citealt{ichimura_chapter_2007} cite minimax bounds in \citealt{stone_optimal_1982}).

In the present setting, the results in this paper show that, even though $\sqrt{n}$ local power is possible in certain special cases, the minimax (worst-case) power is slower than $\sqrt{n}$ when one only places bounds on derivatives of certain objects.  In particular, while a bound on the second derivative of $E(W_i^H|X_i=x)$ and $E(W_i^L|X_i=x)$ does not imply Assumption \ref{smoothness_assump_multi} in the interval regression model, one can construct a dgp such that Assumption \ref{smoothness_assump_multi} holds with $\gamma=2$ for any nonzero bound on the second derivative.  Thus, the minimax rates of local power for CvM statistics under a bound on the second derivative are at least as slow as the rates derived in this paper, which are slower than $\sqrt{n}$.  Since the results in \citet{armstrong_weighted_2014} show that the corresponding KS statistics achieve a better rate for local alternatives uniformly over dgps with a bound on the second derivative (and additional regularity conditions), this means that the KS statistic is preferred to the CvM statistic under a minimax criterion in this class.  See Section \ref{minimax_rates_sec} in the appendix for formal statements.

\section{Local Power Results}\label{loc_power_sec}

This section derives local power results for CvM
test statistics under the conditions given in Section \ref{larger_assumptions_sec}.

\subsection{Instrument Based CvM Statistics with Bounded Weights}\label{inst_bdd_sec}

To describe the power results, we need some additional notation.  Define
\begin{align*}
&\lambda_{\text{bdd}}(a,j,k,p)
=\lambda_{\text{bdd}}(a,\bar m_{\theta,j}(\theta_0,x_k),\psi_{j,k},f_X(x_k),f_\mu(x_k,0),p)  \\
&\equiv \int\int \left| \int
  \left[\|x\|^\gamma \psi_{j,k}\left(\frac{x}{\|x\|}\right)
    + \bar m_{\theta,j}(\theta_0,x_k)a \right]k((x-\tilde x)/h)f_X(x_k)\,
  dx\right|_{-}^p f_\mu(x_k,0)\, d\tilde x\, dh.
\end{align*}

\begin{theorem}\label{bdd_weight_lim_thm}
Let
\begin{align*}
a_n=a n^{-\gamma/\{2[d_X+\gamma+(d_X+1)/p]\}}
\end{align*}
for some vector $a$.
Under Assumptions
\ref{g_kernel_assump},
\ref{mu_assump},
\ref{smoothness_assump_multi},
and
\ref{bdd_y_assump_local},
\begin{align*}
n^{1/2}T_{n,p,1,\mu}(\theta_0+a_n)
\stackrel{p}{\to} \left(\sum_{k=1}^{|\mathcal{X}_0|}\sum_{j\in \tilde J(k)}
  \lambda_{\text{bdd}}(a,j,k,p)\right)^{1/p}
\equiv r_{\text{bdd}}(a)
\end{align*}
where $r_{\text{bdd}}(a)\to 0$ as $a\to 0$.
\end{theorem}

Theorem \ref{bdd_weight_lim_thm} has immediate consequences for the power
of tests based on CvM statistics with bounded weightings.

\begin{theorem}\label{bdd_weight_power_thm}
If, in addition to the conditions of Theorem \ref{bdd_weight_lim_thm},
Assumption \ref{cval_bound_assump} holds, the power
\begin{align*}
E\phi_{n,p,1,\mu}(\theta_0+a_n)
\end{align*}
of the test $\phi_{n,p,1,\mu}(\theta_0+a_n)$
will converge to zero for $r_{\text{bdd}}(a)<c$.
If $a$ is close enough to zero, $r_{\text{bdd}}(a)$ will be less than $c$
so that the power will converge to zero.  If, in
addition, Assumption \ref{cval_limit_assump} holds,
the power
will converge to $1$ for
$r_{\text{bdd}}(a)>c$.
\end{theorem}

The $n^{-\gamma/\{2[d_X+\gamma+(d_X+1)/p]\}}$ rate for instrument based CvM statistics with bounded weights
is slower than
the $n^{-\gamma/\{2[d_X+\gamma]\}}$ rate derived for the corresponding KS test
in Theorem 14 of
\citet{armstrong_asymptotically_2015} (for $\gamma=2$)
and
Theorem 5.1 of
\citet{armstrong_weighted_2014} ($\alpha$ from that paper plays the role
of $\gamma$ here).  Note also that local power increases as $p$ increases,
and becomes aribrarily close to the rate for the KS test as $p$
increases.

\subsection{Instrument Based CvM Statistics with Variance Weights}\label{inst_var_sec}

Define
\begin{align*}
&\lambda_{\text{var}}(a,j,k,p)  \\
&\equiv \int\int \left| \int
  \left[\|x\|^\gamma \psi_{j,k}\left(\frac{x}{\|x\|}\right)
    + \bar m_{\theta,j}(\theta_0,x_k)a\right]
w_j(x_k) h^{-d_X/2}k((x-\tilde x)/h)f_X(x_k)\,
  dx\right|_{-}^p  \\
& f_\mu(x_k,0)\, d\tilde x\, dh
\end{align*}
where
$w_j(x_k)\equiv(s_j^2(x_k,\theta_0)f_X(x_k)\int k(u)^2\, du)^{-1/2}$.

\begin{theorem}\label{var_weight_lim_thm}
Let
\begin{align*}
a_n=a n^{-\gamma/\{2[d_X/2+\gamma+(d_X+1)/p]\}}.
\end{align*}
Suppose that %
$\sigma_n(n/\log n)^{1/2}\to \infty$
and Assumptions
\ref{g_kernel_assump},
\ref{mu_assump},
\ref{smoothness_assump_multi},
and
\ref{bdd_y_assump_local}
hold.  Then
\begin{align*}
n^{1/2}T_{n,p,(\hat \sigma\vee \sigma_n)^{-1},\mu}(\theta_0+a_n)
\le 
\left(
 \sum_{k=1}^{|\mathcal{X}_0|}\sum_{j\in J(k)}
 \lambda_{\text{var}}(a,j,k,p) \right)^{1/p}
+o_p(1)
\equiv r_{\text{var}}(a)+o_p(1)
\end{align*}
where $r_{\text{var}}(a)\to 0$ as $a\to 0$.  If, in addition,
$\sigma_n n^{d_X/\{4[d_X/2+\gamma+(d_X+1)/p]\}}\to 0$,
the above display
will hold with the inequality replaced by equality.
\end{theorem}

The result has immediate consequences for the power of tests based on CvM
statistics with truncated variance weightings.

\begin{theorem}

Let $a_n$ be defined as in Theorem \ref{var_weight_lim_thm} and
suppose that the conditions of that theorem and
Assumption \ref{cval_bound_assump} hold.  The power 
\begin{align*}
E\phi_{n,p,(\sigma\vee\sigma_n)^{-1},\mu}(\theta_0+a_n)
\end{align*}
of the test $\phi_{n,p,(\sigma\vee\sigma_n)^{-1},\mu}(\theta_0+a_n)$
will converge to zero for $r_{\text{var}}(a)<c$.
For $a$ close enough to 0, $r_{\text{var}}(a)$ will be less than $c$ so that the
power will converge to zero.
If, in addition, Assumption \ref{cval_limit_assump} holds and
$\sigma_n n^{d_X/\{4[d_X/2+\gamma+(d_X+1)/p]\}}\to 0$,
the 
power will converge to $1$ for
$r_{\text{var}}(a)>c$.
\end{theorem}

As with bounded weighting functions, the rate for detecting local alternatives with CvM statistics with variance weights is slower than the rate for the corresponding KS test.  The $n^{-\gamma/\{2[d_X/2+\gamma+(d_X+1)/p]\}}$ rate for variance weighted CvM statistics derived above contrasts with the $(n/\log n)^{-\gamma/[2(d_X/2+\gamma)]}$ rate for the corresponding KS test derived in \citet{armstrong_multiscale_2016} and \citet{armstrong_weighted_2014} (the results from the latter paper on rates of convergence of confidence regions in the Hausdorff metric imply these local power results).  The rate for CvM statistics approaches the rate for KS statistics as $p\to \infty$.

\subsection{Statistics Based on Kernel Estimates}\label{kern_sec}

 To describe the results, define
 \begin{align*}
 \lambda_{\text{kern}}(a,h,j,k,p)
 \equiv \int \left| 
 \int \left[\|x\|^\gamma \psi_{j,k}\left(\frac{x}{\|x\|}\right)
      +\bar m_{\theta,j}(\theta_0,x_k)a\right]
   h^{-d_X}k((x-\tilde x)/h) \omega_j(\theta_0,x_k) \, dx
    \right|_{-}^p \, d\tilde x.
 \end{align*}
 and
 \begin{align*}
 \tilde \lambda_{\text{kern}}(a,j,k,p)
 \equiv\int \left|\left[[\|v\|^\gamma \psi_{j,k}\left(\frac{v}{\|v\|}\right)
   +\bar m_{\theta,j}(\theta_0,x_k) a\right]
 \omega_j(\theta_0,x_k) \right|_{-}^p\, dv.
 \end{align*}

 \begin{theorem}\label{kern_lim_thm}
 Suppose that Assumptions
 \ref{mu_assump},
 \ref{smoothness_assump_multi},
 \ref{bdd_y_assump_local} and
 \ref{kern_dens_assump}
 hold, and that the kernel function $k$ satisfies Assumption \ref{g_kernel_assump}.
 In addition, suppose that the bandwidth $h$ satisfies $h/n^{-s}\to c_h$
 for some $0<s<1/d_X$ and $c_h>0$,
 the kernel function $k$ satisfies
 $\int k(u)\, du=1$
 and that the functions $\psi_{j,k}$ in Assumption
 \ref{smoothness_assump_multi} are continuous.
 Let $a_n=an^{-q}$ for some $a\in\mathbb{R}^{d_\theta}$ where
 \begin{align*}
 q=\left\{\begin{array}{cc}
 s\gamma & \text{if }s<1/[2(\gamma+d_X/p+d_X/2)]  \\
 (1-sd_X)/[2(1+d_X/(p\gamma))] & \text{if }s\ge 1/[2(\gamma+d_X/p+d_X/2)]
 \end{array}\right.
 \end{align*}
 and let $\theta_n=\theta_0+a_n$.
 If $s>1/[2(\gamma+d_X/p+d_X/2)]$, then
 \begin{align*}
 (nh^{d_X})^{1/2}T_{n,p,\text{kern}}(\theta_n)
   \stackrel{p}{\to} c_h^{d_X/2}\left(\sum_{k=1}^{|\mathcal{X}_0|}\sum_{j\in J(k)}
      \tilde\lambda_{\text{kern}}(a,j,k,p)\right)^{1/p}
   \equiv \tilde r_{\text{kern}}(a).
 \end{align*}

 If $s=1/[2(\gamma+d_X/p+d_X/2)]$, then
 \begin{align*}
 (nh^{d_X})^{1/2}T_{n,p,\text{kern}}(\theta_n)
   \stackrel{p}{\to} c_h^{d_X/2}\left(\sum_{k=1}^{|\mathcal{X}_0|}\sum_{j\in J(k)}
      \lambda_{\text{kern}}(a,c_h,j,k,p)\right)^{1/p}
 \equiv r_{\text{kern}}(a,c_h).
 \end{align*}

 If $s<1/[2(\gamma+d_X/p+d_X/2)]$, then
 \begin{align*}
 (nh^{d_X})^{1/2}T_{n,p,\text{kern}}(\theta_n)
 \end{align*}
 will converge in probability to $0$ if %
 \begin{align*}
 \left(\sum_{k=1}^{|\mathcal{X}_0|}\sum_{j\in J(k)}\lambda_{\text{kern}}(a,c_h,j,k,p)\right)^{1/p}
 \end{align*}
 is $0$ in a neighborhood of $(a,c_h)$, and will converge to $\infty$ if
 this expression is %
 strictly positive.  %

 \end{theorem}

 The result has immediate implications for the power of tests based on kernel CvM statistics.

 \begin{theorem}
 Let $a_n$ be defined as in Theorem \ref{kern_lim_thm} and
 suppose that the conditions of that theorem and
 Assumption \ref{cval_bound_assump} hold.  If $s>1/[2(\gamma+d_X/p+d_X/2)]$, the power
 \begin{align*}
 E\phi_{n,p,\text{kern}}(\theta_0+a_n)
 \end{align*}
 of the test $\phi_{n,p,\text{kern}}(\theta_0+a_n)$
 will converge to zero for $\tilde r_{\text{kern}}(a)<c$.  If
 $s=1/[2(\gamma+d_X/p+d_X/2)]$, the power given by the above display
 will converge to zero for $\tilde r_{\text{kern}}(a,c_h)<c$.  If
 $s<1/[2(\gamma+d_X/p+d_X/2)]$, the power given by the above display
 will converge to zero if $\tilde r_{\text{kern}}(a,c_h)=0$ in a
 neighborhood of $(a,c_h)$.  If, in addition, Assumption
 \ref{cval_limit_assump} holds, the power given by the above display
 will converge to $1$ if $\tilde r_{\text{kern}}(a)>c$,
 $r_{\text{kern}}(a,c_h)>c$, or $r_{\text{kern}}(a,c_h)>0$ in the cases
 where $s$ is greater than, equal to, or less than
 $1/[2(\gamma+d_X/p+d_X/2)]$ respectively.

 \end{theorem}

 As with instrument based statistics, the rate for detecting local alternatives with the kernel CvM test is slower than the rate for the corresponding KS statistic.  The rate derived in Theorem \ref{kern_lim_thm} can be written as
 $\max\{(nh^{d_X})^{-1/[2(1+d_X/(p\gamma))]},h^{\gamma}\}$, which is slower than the
 $\max\left\{(n h^{d_X}/\log n)^{-1/2},h^{\gamma}\right\}$ rate for kernel based KS statistics derived in \citet{armstrong_weighted_2014}.  As with the instrument based statistics, the CvM test is more powerful for $p$ larger, and the rate approaches the rate for the KS test as $p$ goes to $\infty$.

 Theorem \ref{kern_lim_thm} can be used to choose the optimal bandwidth in this setting.
 The rate $a_n=an^{-q}$ is best when $s=1/[2(\gamma+d_X/p+d_X/2)]$, which gives an exponent in the rate of
 \begin{align*}
 &q=\frac{\gamma}{2(\gamma+d_X/p+d_X/2)}
 =\frac{1-s d_X}{2(1+d_X/(p\gamma))}
 =s\gamma.
 \end{align*}
 Note that this rate is faster than the $n^{-\gamma/[2(d_X/2+\gamma+(d_X+1)/p))]}$ rate that can be obtained with instrument based CvM tests with variance weights.  Thus, restricting the class of instruments using prior knowledge of the data generating process leads to a faster rate with CvM statistics.  In contrast, instrument based KS statistics with variance weights can achieve the same rate as kernel KS statistics that use prior knowledge of the data generating process to choose the bandwidth optimally
\citep[cf.][]{armstrong_weighted_2014,armstrong_multiscale_2016,chetverikov_adaptive_2012}.

 \section{Monte Carlo}\label{monte_carlo_sec}

 This section reports the results of a Monte Carlo study of the finite sample properties of the statistics considered in this paper.
 I perform a Monte Carlo based on a median regression model with potentially endogenously missing data.  I use the same data generating processes as for the Monte Carlo for variance weighted KS statistics in \citet{armstrong_multiscale_2016}.  A description of the model and data generating processes is repeated here for convenience.

 The latent variable $W^*_i$ follows a linear median regression model given the observed covariate $X_i$: $q_{1/2}(W_i^*|X_i)=\theta_1+\theta_2 X_i$ where $q_{1/2}(W_i^*|X_i)$ is the conditional median of $W_i^*$ given $X_i$.  Define $W_i^H=W_i^*$ when $W_i^*$ is observed and $W_i^H=\infty$ otherwise.  This gives the conditional moment inequality $E[I(\theta_1+\theta_2 X_i\le W_i^H)-1/2|X_i]\ge 0$ a.s. (a similar inequality can be formed with the lower bound $W_i^L$ defined analogously, but with $W_i^L=-\infty$ when $W_i^*$ is unobserved, which would give the interval quantile regression setup of Section \ref{int_quant_reg_sec} of the appendix; the Monte Carlo focuses on the inequality corresponding to $W_i^H$ for simplicity).  This model allows for arbitrary correlation between the ``missingness'' process and $(W_i^*,X_i)$, so that the resulting bounds can be used to assess sensitivity to missingness at random assumptions that would point identify the model.

 Each design uses data from the true model $W_i^*=\theta_1^*+\theta_2^* X_i+u_i$, where $(\theta_1^*,\theta_2^*)=(0,0)$ and $u_i$ is independent of $X_i$ with $u_i\sim\text{unif}(-1,1)$.  The outcome variable $W_i^*$ is then set to be missing independently of $W_i^*$ with probability $p(X_i)$ (note that, while the data are generated according to a missingness at random assumption and a particular parameter value, the tests are robust to failure of this assumption, which leads to a lack of point identification), where $p(x)$ is varied in each of three designs:
 \begin{align*}
 \begin{array}{ll}
 \text{Design 1:} & p(x)=.1  \\
 \text{Design 2:} & p(x)=.02+2\cdot .98\cdot |x-.5|  \\
 \text{Design 3:} & p(x)=.02+4\cdot .98\cdot (x-.5)^2.
 \end{array}
 \end{align*}
This leads to the identified set $\Theta_0=\{(\theta_1,\theta_2)'|\theta_1+\theta_2x\le q_{1/2}(W_i^H|X_i=x)\text{ all } x\in [0,1]\}$ where $q_{1/2}(W_i^H|X_i=x)$ can be calculated for each design as $q_{1/2}(W_i^H|X_i=x)=1/(1-p(x))-1$.
 For each design, the Monte Carlo power of $\phi(\theta)$ for each test $\phi$ under the dgp in the given design is reported for
$\theta=(\overline\theta_1+a,0)$ where
 $\overline\theta_1=\sup\{\theta_1|(\theta_1,0)\in\Theta_0\}$ and $a$ varies over the set $\{.1,.2.,.3,.4,.5\}$.  This leads to local alternatives that satisfy the conditions of this paper with $\gamma=1$ for Design 2 and $\gamma=2$ for Design 3.  Design 1 leads to a flat conditional mean for which asymptotic theory predicts the following rates (for the instrument functions used here): $n^{-1/2}$ for kernel and instrument based CvM and unweighted instrument based KS statistics, $(n/\log n)^{-1/2}$ for variance weighted instrument KS statistics and $(nh/\log n)^{-1/2}$ for kernel KS statistics
 \citep[see][]{andrews_inference_2013,armstrong_weighted_2014,chernozhukov_intersection_2013,lee_testing_2013}.

For the instrument based statistics, I use the class of functions $\{x\mapsto I(s<x<s+t)|0\le s\le s+t\le 1\}$ and the the Lebesgue measure on $\{(s,t)|0\le s\le s+t\le 1\}$ for $\mu$ for the instrument based CvM statistics.  This corresponds to the multiscale kernel instruments in Assumption \ref{g_kernel_assump} with the uniform kernel.  For the kernel based statistics, the uniform kernel is used, and the supremum or integral is taken over the set $[h/2,1-h/2]$, so that the support of the kernel function is always contained in the support of $X_i$.
For the CvM statistics, the simulations use the test with $L_p$ exponent $p=1$.
For each test statistic, the critical value is taken from the least favorable null distribution, calculated exactly (up to Monte Carlo error) using the distribution under $(\overline\theta_1,0)$ under Design 1.
For the kernel estimators, the bandwidths $n^{-1/5}$, $n^{-1/3}$ and $n^{-1/2}$ are used, and, for the truncated variance weighted CvM statistics, the values $n^{-1/5}/4$, $n^{-1/3}/4$ and $n^{-1/2}/4$ are used for the
truncation parameter $\sigma_n^2$ (this corresponds to truncating the variance of functions $I(s<x<s+t)$ with $t$ less than $n^{-1/5}$, $n^{-1/3}$ and $n^{-1/2}$).
For comparison, results for the variance weighted instrument KS statistic, which corresponds to the multiscale statistic of \citet{armstrong_multiscale_2016}, are reported as well (taken directly from that paper).

Overall, the Monte Carlo results support the claim that, for the data generating processes and classes of instrument functions considered in the theoretical results in this paper, KS statistics perform better than CvM statistics.  For Design 2 and Design 3, which follow the conditions of this paper with $\gamma=1$ and $\gamma=2$ respectively, the instrument based KS statistic has more power than the instrument based CvM statistic in basically all cases.  For the kernel statistics, the KS test performs better unless the bandwidth is chosen to be much too small.  For example, for Design 3, the optimal bandwidth for the kernel statistic is of order $n^{-1/5}$, and the kernel KS statistic performs better than the kernel CvM statistic with this bandwidth.  However, the kernel statistic performs worse for smaller bandwidths when the sample size is not too large (although the KS statistic does almost as well or better with $1000$ observations, suggesting that the asymptotics of Theorem \ref{kern_lim_thm} have started to kick in at this point).

Note also that power in the Monte Carlo is very sensitive to the design, with greater power for Design 3 than Design 2.  This is to be expected given the asymptotic results.  Under Design 3, the assumptions of this paper hold with $\gamma=2$, while, under Design 2, the assumptions hold with $\gamma=1$.  The results of Section \ref{loc_power_sec} show that asymptotic power is increasing in $\gamma$ (the rate at which local alternatives may approach the null with nontrivial power is faster for larger $\gamma$) for each of the test statistics considered.

For Design 1, asymptotic results from elsewhere in the literature predict that the instrument based statistics with the instruments used here perform about the same (in terms of the rate for detecting local alternatives) for KS and CvM statistics, although the variance weighted KS statistic performs slightly worse (by a $\log n$ factor).  For kernel statistics, asymptotic theory predicts that KS statistics will perform worse than CvM statistics in this case (the latter can achieve a $n^{-1/2}$ rate, while the former cannot if the bandwidth goes to zero).  All of these predictions are borne out in the Monte Carlo: instrument based statistics all perform well with the weighted KS statistics performing slightly worse, while CvM version is better for kernel statistics.

The Monte Carlo results also fit well with the prescription of the weighted instrument KS or ``multiscale'' statistic of \citet{armstrong_weighted_2011}, \citet{armstrong_weighted_2014}, \citet{armstrong_multiscale_2016} and \citet{chetverikov_adaptive_2012} as the only test among the ones considered here that comes close to having the best power among these test statistics for all three Monte Carlo designs (according to asymptotic approximations, the weighted instrument KS test achieves the best rate to at least within a $\log n$ factor in all three cases, while each of the other statistics considered here performs worse by a polynomial factor in at least one case).  While other statistics perform slightly better in certain cases, they perform much worse in others (e.g. the kernel KS statistic performs slightly better in Design 3 with the optimal bandwidth, $n^{-1/5}$, but performs much worse when other bandwidths are chosen, or with any bandwidth choice in Design 1).

\section{Conclusion}\label{conclusion_sec}

This paper derives local power results for tests for conditional moment inequality models based on several forms of CvM statistics in the set identified case.
The power comparisons hold under conditions that arise naturally in the set identified case, and determine the minimax rate.
The results show that KS tests are preferred to CvM statistics and that variance weightings are preferred to bounded weightings.

\appendix

\section{Primitive Conditions and Minimax Bounds}\label{prim_cond_append}

This appendix gives primitive conditions for the assumptions used in this paper, and shows how the (pointwise in the underlying distribution) results for local alternatives considered in the paper can be used to bound the minimax power of CvM tests in classes of underlying distributions where the conditional mean is constrained only by smoothness assumptions.  Since the corresponding KS statistic has a faster rate in these classes, this justifies the claim that the CvM tests considered here perform worse in these models under a minimax criterion.
Section \ref{finite_contact_set_sec} gives general primitive conditions for the assumption that the contact set $\mathcal{X}_0$ in Assumption \ref{smoothness_assump_multi} is finite.
Sections \ref{int_reg_sec}, \ref{int_quant_reg_sec} and \ref{selection_model_sec} provide primitive conditions for the assumptions used in this paper in various settings.
Section \ref{minimax_rates_sec} uses the results in the body of this paper to give conditions under which the CvM statistics considered in this paper do not achieve the optimal rate minimax rate, and verifies these conditions for the interval regression model.

\subsection{Primitive Conditions for Finite Contact Set}\label{finite_contact_set_sec}

If we assume that the support of $X_i$ is compact, and that the minimizing set $\{x|\bar m_j(\theta,x)=0\}$ is contained on the interior of the support of $X_i$, then the minimizing set will be finite so long as $\bar m_j(\theta,x)$ is twice continuously differentiable with strictly positive definite second derivative matrix at any minimum.  This follows from the proof of Lemma B.1 in the supplementary appendix of \citet{armstrong_asymptotically_2015}, and we state the result here for convenience.  (Note that the lemma in \citet{armstrong_asymptotically_2015} assumes a third derivative, since a third derivative is used for other results in that paper.  However, a inspection of the proof shows that a continuous second derivative suffices.)

\begin{lemma}\label{finite_contact_set_lemma}
Let $h:\mathcal{X}\to\mathbb{R}$ be twice continuously differentiable on the compact set $\mathcal{X}\subseteq\mathbb{R}^{k}$.  Suppose that, for any minimizer $\tilde x$ of $h(x)$, $\tilde x$ is on the interior of $\mathcal{X}$, and that the second derivative matrix of $h$ is strictly positive definite at $\tilde x$.  Then the set of minimizers of $h(x)$ over $\mathcal{X}$ is finite.
\end{lemma}
\begin{proof}
The result follows from the proof of Lemma B.1 in the supplementary appendix of \citet{armstrong_asymptotically_2015}.
\end{proof}

\subsection{Interval Regression}\label{int_reg_sec}

This section gives primitive conditions for the interval regression model described in the Introduction, which falls into the setup of this paper with
$W_i=(X_i,W_i^L,W_i^H)$ and $m(W_i,\theta)=(W_i^H-(1,X_i')\theta,(1,X_i')\theta-W_i^L)'$.
First, I prove Theorem \ref{int_reg_second_deriv_thm}.  Then, I give conditions under which the assumptions in the main text hold with $\gamma=1$.

\begin{proof}[Proof of Theorem \ref{int_reg_second_deriv_thm}]
First, note that the set of $x$ such that $\bar m_j(\theta,x)=0$ for some $j$ is finite by Lemma \ref{finite_contact_set_lemma}.
Part (ii) of Assumption \ref{smoothness_assump_multi} follows from a second order Taylor expansion, and part (i) follows by compactness of the support of $X_i$ and continuity of the first two derivatives of the conditional means.
Part (iv) is immediate from part (ii) of the conditions of the theorem
 and the fact that the conditional variance is constant in $\theta$ for this model.
For part (v), note that $\frac{d}{d\theta}\bar m_1(\theta,x)=-\frac{d}{d\theta}\bar m_2(\theta,x)=(1,x')$, which is clearly continuous in $(\theta,x)$.
Assumption \ref{bdd_y_assump_local} is immediate from the bounds on $W_i^H$ and $W_i^L$.
\end{proof}

For the Lipschitz case ($\gamma=1$), we can replace the assumption of two derivatives with a condition on the directional one-sided first derivatives.  Here, we make the assumption of finiteness of the set where the conditional moments bind directly, since arguments involving second derivatives do not apply.  In the following, $\mathbb{S}^{d_X-1}$ denotes the unit sphere $\{u\in\mathbb{R}^{d_X}|\|u\|=1\}$.

\begin{assumption}\label{int_reg_assump_lipschitz}
\begin{itemize}
\item[i.)] The conditional means $E(W_i^H|X_i=x)$ and $E(W_i^L|X_i=x)$ are Lipschitz continuous, $X_i$ has a continuous density and compact support, and $W_i^H$ and $W_i^L$ are bounded from above and below by finite constants.

\item[ii.)]
The set $\mathcal{X}_0\equiv\{x|E(W_i^H|X_i=x)=(1,x')\theta_0\}$ is finite, and,
for any point $\tilde x\in\mathcal{X}_0$, $\tilde x$ is in the interior of the support of $X_i$, $var(W_i^H|X_i=x)$ is positive and continuous at $\tilde x$ and
the one-sided directional derivative
$\frac{d}{dt_+}[E(W_i^H|X_i=\tilde x+tu)-(1,(\tilde x+tu)')\theta_0]$
is bounded from below away from zero at $t=0$
and is right continuous at $t=0$ uniformly over $u\in\mathbb{S}^{d_X-1}$.
The same holds for $E(W_i^L|X_i=x)$ with ``positive'' replaced by ``negative''
in the last statement.
\end{itemize}
\end{assumption}

\begin{theorem}\label{int_reg_lipschitz_thm}
Under Assumption \ref{int_reg_assump_lipschitz}, Assumptions
\ref{smoothness_assump_multi}
and \ref{bdd_y_assump_local} hold, with $\gamma=1$ in Assumption \ref{smoothness_assump_multi}.
\end{theorem}
\begin{proof}
Part (ii) of Assumption \ref{smoothness_assump_multi} follows from a first order Taylor expansion, and part (i) follows by compactness of the support of $X_i$ and the continuity and lower bound on the directional derivatives.
The verification of the remaining conditions is the same as in the twice differentiable case.
\end{proof}

\subsection{Interval Quantile Regression}\label{int_quant_reg_sec}

For the interval quantile regression model, the latent variable $W_i^*$ follows a linear quantile regression model $q_\tau(W_i^*|X_i)=(1,X_i')\theta$, where $\tau$ is given and $q_\tau(U|V)$ denotes the $\tau$th conditional quantile of $U$ given $V$ for random variables $U$ and $V$.  As with interval mean regression, we observe $(X_i,W_i^L,W_i^H)$ where $[W_i^L,W_i^H]$ is known to contain $W_i^*$.  This falls into our setup with
$m(W_i,\theta)=(\tau-I(W_i^H\le (1,X_i')\theta),I(W_i^L\le (1,X_i')\theta)-\tau)'$.

For the interval quantile regression model, one can use essentially the same assumptions as for the interval mean regression model considered above, but with conditional means replaced by conditional quantiles.  In the interest of space, we consider only the case where the conditional quantile function has two derivatives ($\gamma=2$).

\begin{assumption}\label{quant_reg_assump}
\begin{itemize}
\item[i.)] The conditional quantiles $q_\tau(W_i^H|X_i=x)$ and $q_\tau(W_i^L|X_i=x)$ are twice differentiable with continuous second derivatives and $X_i$ has a continuous density and compact support.

\item[ii.)]
For any $\tilde x$ such that $q_\tau(W_i^H|X_i=\tilde x)=(1,\tilde x')\theta_0$,
$\tilde x$ is in the interior of the support of $X_i$ and 
$q_\tau(W_i^H|X_i=x)$ has a positive definite second derivative matrix at $\tilde x$.
The same holds for $q_\tau(W_i^L|X_i=x)$ with ``positive definite'' replaced by ``negative definite.''
\end{itemize}
\end{assumption}

In addition, we will also require an assumption on the conditional densities of $W_i^H$ and $W_i^L$ given $X_i$.

\begin{assumption}\label{W_density_assump}
For some $\eta>0$, $W_i^H|X_i$ and $W_i^L|X_i$ have conditional densities $f_{W_i^H|X_i}(w|x)$ and $f_{W_i^L|X_i}(w|x)$ on $\{(x,w)| q_{\tau,P}(W_i^H|X_i=x)-\eta\le w\le q_{\tau,P}(W_i^H|X_i=x)+\eta\}$ and $\{(x,w)| q_{\tau,P}(W_i^L|X_i=x)-\eta\le w\le q_{\tau,P}(W_i^L|X_i=x)+\eta\}$ respectively
that are continuous as a function of $(x,w)$ and bounded away from zero on these sets.
\end{assumption}

Assumption \ref{W_density_assump} is similar to Assumption B.3 in \citet{armstrong_weighted_2014}.  As discussed in \citet{armstrong_weighted_2014}, this type of condition will hold, for example, when $(X_i,W_i^*)$ has a smooth joint density, and $W_i^*$ is either missing (in which case $W_i^L=-\infty$ and $W_i^H=\infty$) or fully observed (in which case $W_i^L=W_i^H=W_i^*$), so long as the probability that $W_i^*$ is missing conditional on $(X_i,W_i^*)=(x,w)$ is smooth as a function of $(x,w)$.

\begin{theorem}
Suppose that Assumptions \ref{quant_reg_assump} and \ref{W_density_assump} hold.  Then Assumptions \ref{smoothness_assump_multi} and \ref{bdd_y_assump_local} hold, with $\gamma=2$ in Assumption \ref{smoothness_assump_multi}.
\end{theorem}
\begin{proof}
Let $\theta_0\in\Theta_0$ satisfy the conditions of the theorem and let $\tilde x$ be such that $q_{\tau}(W_i^H|X_i=\tilde x)=(1,\tilde x')\theta_0$.  Let $V(x)$ denote the second derivative matrix of $x\mapsto q_{\tau}(W_i^H|X_i= x)$.
Then, for $\delta$ small enough and $\|x-\tilde x\|\le \delta$,
\begin{align*}
&\bar m_1(\theta,x)=\tau-P(W_i^H\le (1,X_i')\theta_0|X_i=x)
=\int_{(1,x')\theta_0}^{q_{\tau}(W_i^H|X_i=x)} f_{W_i^H|X_i}(w|x)\, dw  \\
&=\int_{(1,x')\theta_0}^{(1,x')\theta_0+(x-\tilde x)'V(\tilde x)(x-\tilde x)+r(x)} f_{W_i^H|X_i}(w|x)\, dw
\end{align*}
where $\lim_{x\to \tilde x} r(x)=0$ and the last step follows from a second order Taylor expansion.  This expression is bounded from above by $\overline f(\delta)\cdot [(x-\tilde x)'V(\tilde x)(x-\tilde x)+\overline r(\delta)]$ and from below by $\underline f(\delta)\cdot [(x-\tilde x)'V(\tilde x)(x-\tilde x)+\underline r(\delta)]$ where $\overline f(\delta)$ and $\overline r(\delta)$ are upper bounds for $f_{W_i^H|X_i}(w|x)$ and $r(x)$ on $\{(x,w)| \|x-\tilde x\|\le \delta, (1,x')\theta_0\le w\le q_\tau(W_i^H|X_i=x)\}$ and $\underline f(\delta)$ and $\underline r(\delta)$ are lower bounds.  As $\delta\to 0$, $\overline f(\delta)$ and $\underline f(\delta)$ converge to $f_{W_i^H|X_i}((1,\tilde x')\theta_0|\tilde x)$ and $\overline r(\delta)$ and $\underline r(\delta)$ converge to $0$, so that
\begin{align*}
\sup_{\|x-\tilde x\|\le \delta}
\left\|\frac{\tau-P(W_i^H\le (1,X_i')\theta_0|X_i=x)}
  {\|x-\tilde x\|^2}
-\frac{(x-\tilde x)'}{\|x-\tilde x\|}V(\tilde x)\frac{(x-\tilde x)'}{\|x-\tilde x\|}\cdot f_{W_i^H|X_i}((1,\tilde x')\theta_0|\tilde x)\right\|
\stackrel{\delta\to 0}{\to} 0.
\end{align*}
Applying this argument to the finite set of values $\tilde x$ such that $\tau-P(W_i^H\le (1,X_i')\theta_0|X_i=x)=0$ and a symmetric argument for $W_i^L$, it follows that part (ii) of Assumption \ref{smoothness_assump_multi} holds with $\gamma=2$.

To verify part (i) of Assumption \ref{smoothness_assump_multi} first note that the set $\mathcal{X}_0=\{x | q_\tau(W_i^H|X_i=x)=(1,x')\theta\}$ is finite by Lemma \ref{finite_contact_set_lemma}.  Using this and similar arguments to those used in the proof of Theorem \ref{int_reg_second_deriv_thm}, there exists $\varepsilon>0$ and $\delta>0$ such that
 $q_{\tau}(W_i^H|X_i=x)-(1,x)'\theta$ is bounded away from zero for $\|\theta-\theta_0\|<\varepsilon$
and $x$ such that, for all $\tilde x\in\mathcal{X}_0$, $\|x-\tilde x\|\ge \delta$.
It then follows from Assumption \ref{W_density_assump} that $\tau-P(W_i^H\le (1,X_i')\theta_0|X_i=x)$ is bounded away from zero on such a set.  Part (i) of Assumption \ref{smoothness_assump_multi} follows from this and a similar argument for $W_i^L$.

For part (iv) of Assumption \ref{smoothness_assump_multi}, note that the conditional variance of the moment function corresponding to $W_i^H$ is $P(W_i^H\le (1,x')\theta|X_i=x)[1-P(W_i^H\le (1,x')\theta|X_i=x)]$, so it suffices to show that $P(W_i^H\le (1,x')\theta|X_i=x)$ is in the set $(0,1)$ and is continuous in $(\theta,x)$ at each $(\theta_0,\tilde x)$ such that $\bar m_1(\theta,x)=P(W_i^H\le (1,\tilde x')\theta_0|X_i=\tilde x)=\tau$.  This follows since, by Assumption \ref{W_density_assump}, $W_i^H$ has a continuous conditional density in a neighborhood of $(1,\tilde x')\theta_0$.

For part (v) of Assumption \ref{smoothness_assump_multi}, note that, for $(x,\theta)$ such that $W_i^H$ has a conditional density given $X_i=x$ at $(1,x')\theta$,
\begin{align*}
\bar m_{\theta,1}(\theta,x)=-\frac{d}{d\theta'} P(W_i^H\le (1,x')\theta|X_i=x)
=-f_{W_i^H|X_i=x}((1,x')\theta|x)(1,x').
\end{align*}
This is continuous in $(\theta,x)$ in a small enough neighborhood of any $(\theta_0,\tilde x)$ with $\bar m_{\theta,1}(\theta_0,\tilde x)=0$, since $f_{W_i^H|X_i=x}(w|x)$ is continuous for $w$, $x$ in a neighborhood of
at $x=\tilde x$ and $w=(1,\tilde x')\theta_0$ for any such $\theta_0$ and $\tilde x$ by Assumption \ref{W_density_assump}.

\end{proof}

\subsection{Selection Model}\label{selection_model_sec}

The interval regression model contains, as a special case, an approach to selection models based on bounds suggested in \citet{manski_nonparametric_1990}.
In particular, consider a selection model in which we are interested in the mean of $Y_i^*$, which is not always observed.
Suppose that $Y_i^*$ is known to take values in $[\underline Y, \overline Y$] for some fixed $\underline Y$ and $\overline Y$, and a variable $X_i$ is available such that $E(Y_i^*|X_i)=E(Y_i^*)$ (i.e. $Y_i^*$ is mean independent of $X_i$), and such that $X_i$ shifts the conditional probability of observing $Y_i^*$.  For example, we may be interested in the offer wage $Y_i^*$, which is typically only observed when individual $i$ actually works.  In this case, the variable $X_i$ can be taken to be anything that shifts labor force participation through the opportunity cost of working (such as income from other sources such as family or government benefits) while being independent of the distribution of offer wages.

Let $D_i$ denote an indicator variable that is $1$ when $Y_i^*$ is observed and $0$ otherwise.  We observe $(X_i,Y_i,D_i)$ where $Y_i=D_i\cdot Y_i^*$.  Following \citet{manski_nonparametric_1990}, note that, letting $W_i^L=Y_i\cdot D_i+\underline Y\cdot (1-D_i)$ and $W_i^L=Y_i\cdot D_i+\overline Y\cdot (1-D_i)$, we have $W_i^L\le Y_i^*\le W_i^H$ with probability one.  Letting $\theta=E(Y_i^*)$ and using the fact that $E(Y_i^*)=E(Y_i^*|X_i)$ a.s., we obtain our setup with
$m(W_i,X_i,\theta)=(W_i^H-\theta,\theta-W_i^L)'$.  This is a special case of the interval regression model of Section \ref{int_reg_sec}, with $(\theta,0_{1\times d_X})$ playing the role of $\theta$.  That is, we have the interval regression model with the slope parameter constrained to be zero.  Thus, if we consider a null value $\theta_0$ and a sequence of alternatives in the interval regression model for which the slope parameter is zero, the results of Section \ref{int_reg_sec} apply immediately to give primitive conditions for Assumption \ref{smoothness_assump_multi} (here Assumption \ref{bdd_y_assump_local} holds by construction and the assumption that $Y_i^*$ is bounded).

Note that
$E(W_i^H|X_i=x)=E(Y_i^*D_i|X_i=x)+\overline Y\cdot [1-P(D_i=1|X_i=x)]$.
Thus, a sufficient condition for $E(W_i^H|X_i=x)$ to be twice differentiable (or Lipschitz) is for $P(D_i=1|X_i=x)$ and $E(Y_i^*D_i|X_i=x)$ to be twice differentiable (or Lipschitz).
It is also worth noting that cases where $E(W_i^H|X_i=x)$ is minimized at the (possibly infinite) boundary of the support of $X_i$ are often of interest, and arise naturally in this setting (see, e.g., \citealt{andrews_semiparametric_1998} and \citealt{heckman_varieties_1990}).
While Assumption \ref{smoothness_assump_multi} formally precludes the possibility that the minimum of $E(W_i^H|X_i=x)$ is taken at the boundary of the support of $X_i$, such cases can be handled for certain forms of instrument based statistics by transforming the support of $X_i$ (see Section B.3 of \citealt{armstrong_weighted_2014} for an example of this type of argument applied to instrument based KS statistics).  We leave this extension for future research.

\subsection{Minimax Rates}\label{minimax_rates_sec}

The power results in this paper hold under conditions that are arguably common in practice in the set identified case.  However, there are certainly cases (data generating processes, points on the boundary of the identified set and directions for the local alternative) for which other conditions will be appropriate.  The purpose of this section is to show that, if the underlying distribution is constrained only by smoothness conditions and other regularity conditions, there will always exist a possible underlying distribution and sequence of local alternatives that satisfy these properties, with $\gamma$ governed by the smoothness conditions imposed.  Thus, any test that achieves good uniform power in these classes against alternatives that are closer than the pointwise rates derived here for CvM statistics will be preferred under a minimax criterion.
By results in \citet{armstrong_weighted_2014}, it follows that,
for certain classes of alternatives defined by smoothness conditions, the variance weighted KS statistic of \citet{armstrong_weighted_2014}, \citet{armstrong_multiscale_2016} and \citet{chetverikov_adaptive_2012} is preferred to the CvM statistics considered in this paper under a minimax criterion.

To formalize these ideas, the rest of this section considers classes $\mathcal{P}$ of underlying distributions and uses the notation $E_P$ and $\Theta_0(P)$ to denote expectations and the identified set under a distribution $P$.
In the results below, $d(\theta,\tilde\theta)$ denotes the Euclidean distance $\|\theta-\tilde\theta\|$.

\begin{theorem}\label{minimax_cvm_thm_highlevel}
Let $\phi_{CvM}(\theta)$ be one of the CvM tests defined in (\ref{test_def_rootn}) or (\ref{test_def_kern}) with the critical value satisfying Assumption \ref{cval_bound_assump}, the class $\mathcal{G}$ or kernel function $k$ satisfying Assumption \ref{g_kernel_assump}, and the measure $\mu$ satisfying Assumption \ref{mu_assump} for the instrument case and the weighting satisfying Assumption
\ref{kern_dens_assump} for the kernel case.
Let $\mathcal{P}$ be any class of distributions such that, for some $P^*\in\mathcal{P}$ and $\theta_0^*$ on the boundary of $\Theta_0(P^*)$, Assumptions \ref{smoothness_assump_multi}
and \ref{bdd_y_assump_local} hold, and either (a) $\theta_0^*$ is on the boundary of the convex hull of $\Theta_0(P^*)$ or (b) for some $a\in\mathbb{R}^{d_\theta}$ and a constant $K$, $d(\theta_0^*,\theta_0^*+a r)\le K\cdot d(\theta_0,\theta_0^*+a r)$ for all $\theta_0\in\Theta_0(P^*)$ and $r$ small enough.
Then, for a small enough constant $C_*>0$,
\begin{align*}
\limsup_{n\to\infty}\inf_{P\in\mathcal{P}}\inf_{\theta \text{ s.t. } d(\theta,\theta_0)\ge C_* r_n \text{ all } \theta_0\in\Theta_0(P)}
E_P\phi_{CvM}(\theta)=0,
\end{align*}
where $r_n$ is the rate for the given test in Section \ref{loc_power_sec} with $\gamma$ given in Assumption \ref{smoothness_assump_multi}.
\end{theorem}
\begin{proof}
Under condition (b), the result is immediate from the results in the main text, since the quantity in the display in the theorem is less than
$\limsup_{n\to\infty} E_{P^*}\phi_{CvM}(\theta_0^*+a C_* r_n K/\|a\|)$ for $P^*$, $\theta_0^*$ and $a$ given in the theorem.
The result follows since
condition (a) implies condition (b) with $K=1$.
To see this, note that, by the supporting hyperplane theorem, there exists a vector $a$ with $\|a\|=1$ such that $a'\tilde\theta_0\le a'\theta_0^*$ for all $\tilde\theta_0$ in the convex hull of $\Theta_0(P^*)$.  For this $a$ and any scalar $r>0$ and $\tilde\theta_0\in\Theta_0(P^*)$,
$d(\theta_0^*+a r,\tilde \theta_0)^2-d(\theta_0^*+a r,\theta_0)^2
=\|\theta_0^*+a r-\tilde \theta_0\|^2-r^2a'a
=\|\theta_0^*-\tilde \theta_0\|^2+2ra'(\theta_0^*-\tilde\theta_0)+r^2 a'a-r^2a'a
\ge \|\theta_0^*-\tilde \theta_0\|^2\ge 0$.
\end{proof}

A class $\mathcal{P}$ of underlying distributions will typically contain a $P^*$ satisfying these conditions so long as it is sufficiently unrestricted (e.g. if the only restrictions are smoothness conditions, etc.).  Theorems \ref{lipschitz_minimax_thm} and \ref{second_deriv_minimax_thm} below give primitive conditions for this in the interval regression model.

Under additional regularity conditions on $\mathcal{P}$, the inverse variance weighted KS statistic of \citet{armstrong_weighted_2014}, \citet{armstrong_multiscale_2016} and \citet{chetverikov_adaptive_2012} achieves a strictly better minimax rate than the upper bounds for CvM statistics given in Theorem \ref{minimax_cvm_thm_highlevel}.  This is stated in the next theorem, which follows immediately from results in \citet{armstrong_weighted_2014} (the results in \citealp{armstrong_weighted_2014} consider a stronger notion of coverage and power).

For concreteness, let us consider a specific version of the inverse variance weighted KS statistic considered in \citet{armstrong_weighted_2014}.  Let $T_{n,\infty,(\sigma\vee \sigma_n)^{-1}}(\theta)$ be given by (\ref{iv_stat_ks_eq}) with $\mathcal{G}=\{x\mapsto I(\|x-\tilde x\|\le h)|\tilde x\in\mathbb{R}^{d_X}, h\in[0,\infty)\}$ and $\omega_j(\theta,g)=\{\hat\sigma_j(\theta,g)\vee [(\log n)^2/n]\}^{-1}$.  Let $\phi_{n,\infty,(\sigma\vee \sigma_n)^{-1}}(\theta)$ be given by (\ref{test_def_rootlogn})
with this definition of $T_{n,\infty,(\sigma\vee \sigma_n)^{-1}}(\theta)$ and
with $\hat c_{n,\infty,(\sigma\vee \sigma_n)^{-1}}$ given by the constant $K$ in Theorem 3.1 in \citet{armstrong_weighted_2014}.  In the interest of concreteness, the above formulation uses certain conservative constants and tuning parameters in defining the test $\phi_{n,\infty,(\sigma\vee \sigma_n)^{-1}}(\theta)$.  Less conservative and data driven methods for choosing these constants have been considered by \citet{armstrong_multiscale_2016} and \citet{chetverikov_adaptive_2012}.

\begin{theorem}\label{minimax_ks_thm_highlevel}
\begin{sloppypar}
Suppose that $\mathcal{P}$ satisfies Assumptions 4.1, 4.3, 4.4 and 4.5
in \citet{armstrong_weighted_2014}, with $\gamma$ taking the place of $\alpha$ in that paper.
Then
$\limsup_{n\to\infty}\sup_{P\in\mathcal{P}}\sup_{\theta_0\in\Theta_0(P)}
E_P\phi_{n,\infty,(\sigma\vee \sigma_n)^{-1}}(\theta_0)=0$
and, for a large enough constant $C^*$,
\begin{align*}
\liminf_{n\to\infty}\inf_{P\in\mathcal{P}}\inf_{\theta \text{ s.t. } d(\theta,\theta_0)\ge C^* [(\log n)/n]^{\gamma/(d_X+2\gamma)} \text{ all } \theta_0\in\Theta_0(P)}
E_P\phi_{n,\infty,(\sigma\vee \sigma_n)^{-1}}(\theta)=1.
\end{align*}
\end{sloppypar}
\end{theorem}
\begin{proof}

Since Assumptions 3.1-3.3 in \citet{armstrong_weighted_2014} follow by definition of the statistic, the result follows from Theorem 4.2 in that paper, with Assumption 4.2(i) in \citet{armstrong_weighted_2014} following from Theorem 4.3 in that paper (since Assumption 4.6 and 4.2(ii) in that paper hold by construction).  For $\mathcal{C}_n$ the setwise confidence set constructed from $\phi_{n,\infty,(\sigma\vee \sigma_n)^{-1}}(\theta)$ in \citet{armstrong_weighted_2014},
\begin{align*}
&\inf_{P\in\mathcal{P}}\inf_{\theta \text{ s.t. } d(\theta,\theta_0)\ge C^* [(\log n)/n]^{\gamma/(d_X+2\gamma)} \text{ all } \theta_0\in\Theta_0(P)}
E_P\phi_{n,\infty,(\sigma\vee \sigma_n)^{-1}}(\theta)  \\
&=\inf_{P\in\mathcal{P}}\inf_{\theta \text{ s.t. } d(\theta,\theta_0)\ge C^* [(\log n)/n]^{\gamma/(d_X+2\gamma)} \text{ all } \theta_0\in\Theta_0(P)}
P(\theta\not\in \mathcal{C}_n)  \\
&\ge \inf_{P\in\mathcal{P}}
P(\theta\not\in \mathcal{C}_n \text{ all }
{\theta \text{ s.t. } d(\theta,\theta_0)\ge C^* [(\log n)/n]^{\gamma/(d_X+2\gamma)} \text{ all } \theta_0\in\Theta_0(P)})  \\
&\ge \inf_{P\in\mathcal{P}}P(d_H(\Theta_0(P),\mathcal{C}_n)< C^* [(\log n)/n]^{\gamma/(d_X+2\gamma)})
\end{align*}
where $d_H(A,B)=\max\{\sup_{a\in A}\inf_{b\in B}d(a,b),\sup_{b\in B}\inf_{a\in A}d(a,b)\}$ is the Hausdorff distance.
This converges to $1$ for large enough $C^*$ by Theorem 4.2 in \citet{armstrong_weighted_2014}.
\end{proof}

The classes $\mathcal{P}$ used in Theorem \ref{minimax_ks_thm_highlevel} impose smoothness conditions on the conditional mean along with a condition on the derivative of the conditional mean with respect to $\theta$ (cases where the latter condition fails appear to favor KS statistics over CvM statistics as well; see Section A.4 of \citealp{armstrong_weighted_2014}).
Note that the rate
given above for the weighted KS statistic $\phi_{n,\infty,(\sigma\vee\sigma_n)^{-1}}$
corresponds to the minimax $L_\infty$ rate for nonparametric testing problems
\citep{lepski_asymptotically_2000} and to the minimax rate for estimating a conditional mean (\citealp{stone_optimal_1982}; see \citealp{menzel_consistent_2010} for related results for estimating the identified set in a setting similar to the one considered here).  The results here show that the CvM statistics considered here do not achieve this rate, and in fact have a minimax rate that is worse by at least a polynomial amount.

I now turn to the interval regression model and consider primitive conditions.
The next two theorems show that certain classes of underlying distributions for the interval regression model will always contain a distribution with a sequence of local alternatives that satisfy the conditions of this paper.
The conclusion of Theorem \ref{minimax_cvm_thm_highlevel} then follows immediately, since the identified set is convex in the interval regression model.
Theorem \ref{lipschitz_minimax_thm} considers the case where the constraints on the conditional mean embodied in $\mathcal{P}$ essentially only restrict the conditional means of $W_i^H$ and $W_i^L$ to a Lipschitz smoothness class.  Theorem \ref{second_deriv_minimax_thm} considers the smoother case where a bound is placed on the second derivative.
For primitive conditions for the conditions of Theorem \ref{minimax_ks_thm_highlevel} in the interval regression model for the case where $d_X=1$ and $\gamma=1$ or $2$, see \citet{armstrong_weighted_2014}, Section 6.2.

\begin{theorem}\label{lipschitz_minimax_thm}
Let $\mathcal{P}$ be any class of underlying distributions for $(X_i,W_i^H,W_i^L)$ in the interval regression model such that,
for all $P\in\mathcal{P}$,
$W_i^H$ and $W_i^L$ are bounded and
$X_i$ has a continuous density on its support $\mathcal{X}_P$.
Suppose that, for some set $\mathcal{X}\subseteq\mathbb{R}^{d_X}$ and some interval $[a,b]$, the following holds:
for any function $f:\mathcal{X}\to [a,b]$ such that
  \begin{align*}
    |f(x)-f(\tilde x)|\le K\|x-\tilde x\|,
  \end{align*}
there exists a $P\in\mathcal{P}$ such that $E_P(W_i^H|X_i)=f(X_i)$ and $E_P(W_i^L|X_i)\le a$ almost surely, and $\mathcal{X}_P=\mathcal{X}$.
Then there exists a $P^*\in\mathcal{P}$ and $\theta_0^*\in\Theta_0(P^*)$ that satisfies
the conditions of Theorem \ref{minimax_cvm_thm_highlevel},
with $\gamma=1$ and $\psi_{j,k}(u)=K$ in Assumption \ref{smoothness_assump_multi}.
\end{theorem}
\begin{proof}
Under these assumptions, there exists a distribution $P\in\mathcal{P}$ such that
$E_P(W_i^H|X_i=x)=b-K[(\varepsilon-\|x-x_0\|)\vee 0]$ for some $\varepsilon>0$ and $x_0$ on the interior of the support of $X_i$,
and $E_P(W_i^L|X_i=x)$ is bounded from above away from $b-2\varepsilon$.  For $\theta=(b-K\varepsilon,0)$, this satisfies the conditions of Theorem \ref{int_reg_lipschitz_thm}.
\end{proof}

\begin{theorem}\label{second_deriv_minimax_thm}
Let $\mathcal{P}$ be any class of underlying distributions for $(X_i,W_i^H,W_i^L)$ in the interval regression model such that,
for all $P\in\mathcal{P}$,
$W_i^H$ and $W_i^L$ are bounded and
$X_i$ has a continuous density on its support $\mathcal{X}_P$.
Suppose that, for some set $\mathcal{X}\subseteq\mathbb{R}^{d_X}$ and some interval $[a,b]$,
for any function $f:\mathcal{X}\to [a,b]$ such that
\begin{align*}
\left|\frac{d^2}{dt^2}f(x+tu)\right|\le K
\end{align*}
for all $u\in\mathbb{R}^{d_X}$ with $\|u\|=1$,
there exists a $P\in\mathcal{P}$ such that $E_P(W_i^H|X_i)=f(X_i)$ and $E_P(W_i^L|X_i)\le a$ almost surely, and $\mathcal{X}_P=\mathcal{X}$.
Then there exists a $P^*\in\mathcal{P}$ and $\theta_0^*\in\Theta_0(P^*)$ that satisfies
the conditions of Theorem \ref{minimax_cvm_thm_highlevel},
with $\gamma=2$ and $\psi_{j,k}(u)=K/2$ in Assumption \ref{smoothness_assump_multi}.
\end{theorem}
\begin{proof}
The result follows by similar arguments to Theorem \ref{lipschitz_minimax_thm} since a function can be constructed for $E_P(W_i^H|X_i=x)$ that has a unique interior minimum with second derivative matrix $K I$ at its minimum and takes values between, say, $(a+b)/2$ and $b$.
\end{proof}

\bibliography{library}

\clearpage

\begin{figure}[h]
  \centering
\includegraphics[height=3.5in]{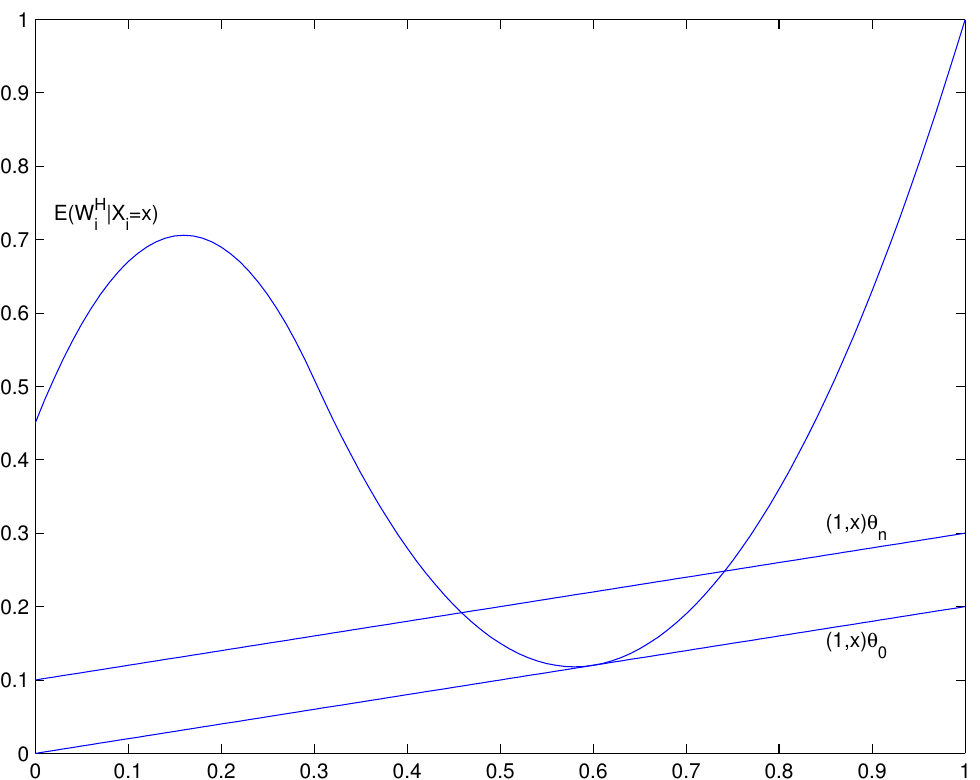}
\caption{Local Alternative for Interval Regression Model}\label{local_alt_fig}
\end{figure}

\begin{figure}[h]
  \centering
\includegraphics[height=3.5in]{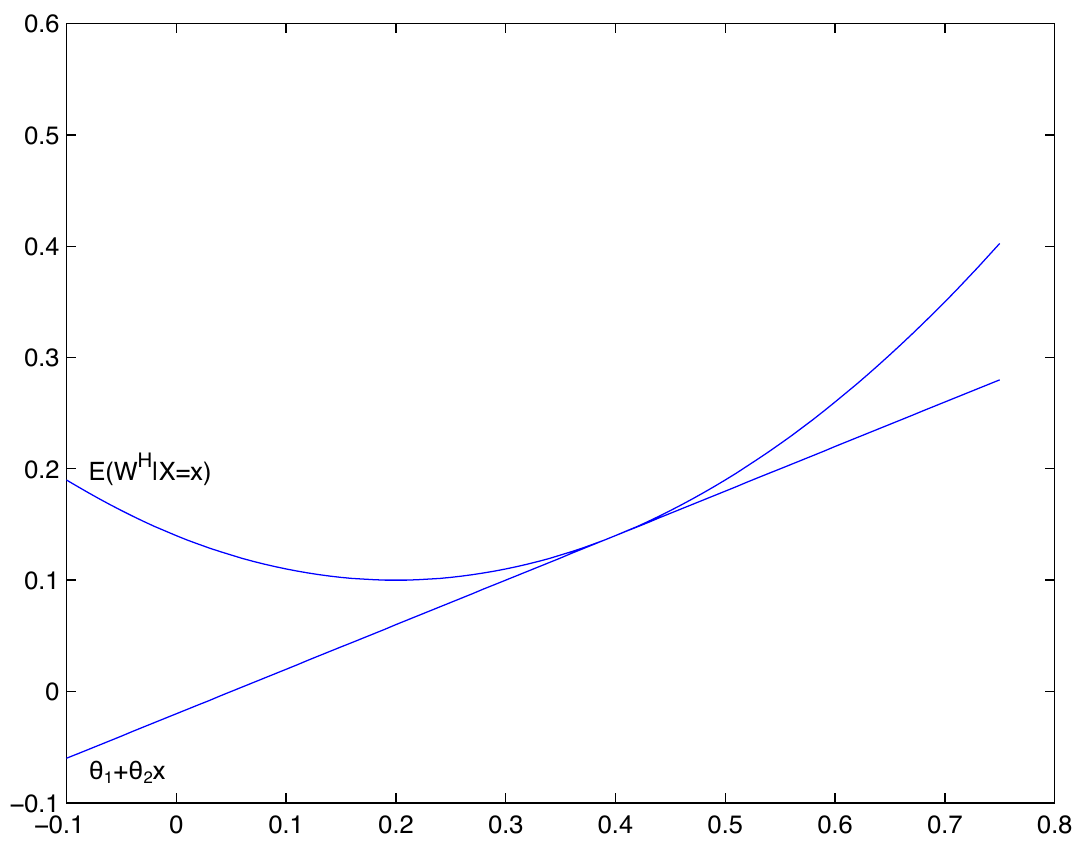}
\caption{Case where Assumption \ref{smoothness_assump_multi} holds with $\gamma=2$}\label{int_reg_smooth_fig}

\vspace{.5in}

  \centering  \includegraphics[height=3.5in]{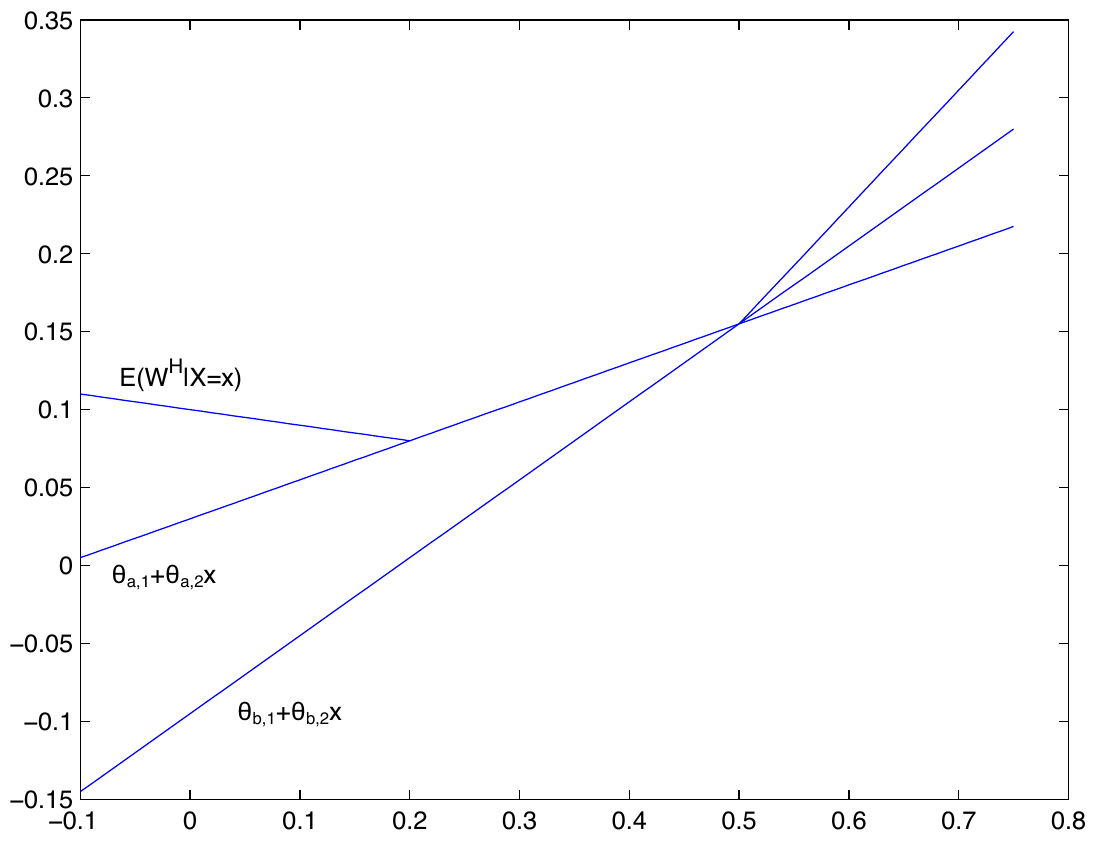}  \caption{Case where Assumption \ref{smoothness_assump_multi} does not hold ($\theta_a$) and
  case where Assumption \ref{smoothness_assump_multi} holds with $\gamma=1$ ($\theta_b$)}\label{int_reg_rootn_fig}
\end{figure}

\clearpage

\begin{table}
\centering
\begin{tabular}{l|l|l|l}
$\theta_1-\overline\theta_1$&$n=100$&$n=500$&$n=1000$\\\hline
0.1&0.196&0.593&0.818\\
0.2&0.458&0.973&1\\
0.3&0.775&1&1\\
0.4&0.952&1&1\\
0.5&0.995&1&1
\end{tabular}

\caption{Power for Unweighted Instrument CvM Test under Design 1}
\label{uw_cvm_power_d1}
\end{table}

\begin{table}
\centering
\begin{tabular}{l|l|l|l}
$\theta_1-\overline\theta_1$&$n=100$&$n=500$&$n=1000$\\\hline
0.1&0.166&0.644&0.835\\
0.2&0.442&0.989&1\\
0.3&0.781&1&1\\
0.4&0.957&1&1\\
0.5&0.994&1&1
\end{tabular}

\caption{Power for Unweighted Instrument KS Test under Design 1}
\label{uw_ks_power_d1}
\end{table}

\begin{table}
\centering
\begin{tabular}{l|l|l|l|l}
$\sigma_n^2$&$\theta_1-\overline\theta_1$&$n=100$&$n=500$&$n=1000$\\\hline
&0.1&0.198&0.567&0.859\\
&0.2&0.49&0.977&1\\
$\frac{1}{4}n^{-1/5}$&0.3&0.77&1&1\\
&0.4&0.955&1&1\\
&0.5&0.997&1&1\\\hline
&0.1&0.208&0.62&0.851\\
&0.2&0.475&0.983&1\\
$\frac{1}{4}n^{-1/3}$&0.3&0.808&1&1\\
&0.4&0.958&1&1\\
&0.5&0.994&1&1\\\hline
&0.1&0.203&0.591&0.822\\
&0.2&0.474&0.981&1\\
$\frac{1}{4}n^{-1/2}$&0.3&0.804&1&1\\
&0.4&0.946&1&1\\
&0.5&0.996&1&1
\end{tabular}

\caption{Power for Weighted Instrument CvM Test under Design 1}
\label{w_cvm_power_d1}
\end{table}

\begin{table}
\centering
\begin{tabular}{l|l|l|l|l}
$t_n$&$\theta_1-\overline\theta_1$&$n=100$&$n=500$&$n=1000$\\\hline
&0.1&0.207&0.503&0.729\\
&0.2&0.48&0.954&1\\
$n^{-1/5}$&0.3&0.759&1&1\\
&0.4&0.956&1&1\\
&0.5&0.997&1&1\\\hline
&0.1&0.144&0.453&0.63\\
&0.2&0.378&0.939&0.998\\
$n^{-1/3}$&0.3&0.691&1&1\\
&0.4&0.886&1&1\\
&0.5&0.982&1&1\\\hline
&0.1&0.156&0.358&0.502\\
&0.2&0.348&0.898&0.991\\
$n^{-1/2}$&0.3&0.649&0.999&1\\
&0.4&0.862&1&1\\
&0.5&0.974&1&1
\end{tabular}

\caption{Power for Weighted Instrument KS Test under Design 1 (from \citet{armstrong_multiscale_2016})}
\label{w_ks_power_d1}
\end{table}

\begin{table}
\centering
\begin{tabular}{l|l|l|l|l}
$h_n$&$\theta_1-\overline\theta_1$&$n=100$&$n=500$&$n=1000$\\\hline
&0.1&0.186&0.547&0.858\\
&0.2&0.453&0.97&1\\
$n^{-1/5}$&0.3&0.729&1&1\\
&0.4&0.934&1&1\\
&0.5&0.994&1&1\\\hline
&0.1&0.188&0.663&0.843\\
&0.2&0.452&0.987&1\\
$n^{-1/3}$&0.3&0.794&1&1\\
&0.4&0.947&1&1\\
&0.5&0.997&1&1\\\hline
&0.1&0.185&0.582&0.848\\
&0.2&0.443&0.977&1\\
$n^{-1/2}$&0.3&0.78&1&1\\
&0.4&0.942&1&1\\
&0.5&0.997&1&1
\end{tabular}

\caption{Power for Kernel CvM Test under Design 1}
\label{kern_cvm_power_d1}
\end{table}

\begin{table}
\centering
\begin{tabular}{l|l|l|l|l}
$h_n$&$\theta_1-\overline\theta_1$&$n=100$&$n=500$&$n=1000$\\\hline
&0.1&0.16&0.439&0.625\\
&0.2&0.343&0.92&0.997\\
$n^{-1/5}$&0.3&0.62&0.999&1\\
&0.4&0.883&1&1\\
&0.5&0.975&1&1\\\hline
&0.1&0.095&0.266&0.481\\
&0.2&0.201&0.715&0.929\\
$n^{-1/3}$&0.3&0.382&0.976&1\\
&0.4&0.606&0.999&1\\
&0.5&0.809&1&1\\\hline
&0.1&0&0.094&0.138\\
&0.2&0&0.255&0.404\\
$n^{-1/2}$&0.3&0&0.508&0.773\\
&0.4&0&0.812&0.982\\
&0.5&0&0.976&1
\end{tabular}

\caption{Power for Kernel KS Test under Design 1}
\label{kern_ks_power_d1}
\end{table}

\begin{table}
\centering
\begin{tabular}{l|l|l|l}
$\theta_1-\overline\theta_1$&$n=100$&$n=500$&$n=1000$\\\hline
0.1&0&0&0\\
0.2&0.001&0&0\\
0.3&0.005&0&0\\
0.4&0.008&0.001&0.004\\
0.5&0.023&0.054&0.119
\end{tabular}

\caption{Power for Unweighted Instrument CvM Test under Design 2}
\label{uw_cvm_power_d2}
\end{table}

\begin{table}
\centering
\begin{tabular}{l|l|l|l}
$\theta_1-\overline\theta_1$&$n=100$&$n=500$&$n=1000$\\\hline
0.1&0&0&0\\
0.2&0.003&0.002&0.001\\
0.3&0.007&0.022&0.037\\
0.4&0.01&0.145&0.412\\
0.5&0.039&0.596&0.884
\end{tabular}

\caption{Power for Unweighted Instrument KS Test under Design 2}
\label{uw_ks_power_d2}
\end{table}

\begin{table}
\centering
\begin{tabular}{l|l|l|l|l}
$\sigma_n^2$&$\theta_1-\overline\theta_1$&$n=100$&$n=500$&$n=1000$\\\hline
&0.1&0&0&0\\
&0.2&0&0&0\\
$\frac{1}{4}n^{-1/5}$&0.3&0.003&0&0\\
&0.4&0.007&0.006&0.013\\
&0.5&0.04&0.118&0.294\\\hline
&0.1&0&0&0\\
&0.2&0&0&0\\
$\frac{1}{4}n^{-1/3}$&0.3&0.001&0.001&0\\
&0.4&0.011&0.009&0.016\\
&0.5&0.032&0.139&0.371\\\hline
&0.1&0&0&0\\
&0.2&0.001&0&0\\
$\frac{1}{4}n^{-1/2}$&0.3&0.003&0&0\\
&0.4&0.009&0.003&0.014\\
&0.5&0.034&0.114&0.288
\end{tabular}

\caption{Power for Weighted Instrument CvM Test under Design 2}
\label{w_cvm_power_d2}
\end{table}

\begin{table}
\centering
\begin{tabular}{l|l|l|l|l}
$t_n$&$\theta_1-\overline\theta_1$&$n=100$&$n=500$&$n=1000$\\\hline
&0.1&0&0&0\\
&0.2&0.006&0.016&0.032\\
$n^{-1/5}$&0.3&0.026&0.138&0.295\\
&0.4&0.064&0.449&0.831\\
&0.5&0.175&0.848&0.995\\\hline
&0.1&0.007&0.012&0.005\\
&0.2&0.016&0.062&0.1\\
$n^{-1/3}$&0.3&0.041&0.215&0.456\\
&0.4&0.119&0.604&0.876\\
&0.5&0.21&0.902&0.996\\\hline
&0.1&0.006&0.014&0.01\\
&0.2&0.023&0.057&0.086\\
$n^{-1/2}$&0.3&0.038&0.229&0.389\\
&0.4&0.119&0.532&0.791\\
&0.5&0.203&0.85&0.982
\end{tabular}

\caption{Power for Weighted Instrument KS Test under Design 2 (from \citet{armstrong_multiscale_2016})}
\label{w_ks_power_d2}
\end{table}

\begin{table}
\centering
\begin{tabular}{l|l|l|l|l}
$h_n$&$\theta_1-\overline\theta_1$&$n=100$&$n=500$&$n=1000$\\\hline
&0.1&0&0&0\\
&0.2&0.001&0.002&0\\
$n^{-1/5}$&0.3&0.008&0.007&0.024\\
&0.4&0.012&0.108&0.369\\
&0.5&0.074&0.484&0.923\\\hline
&0.1&0&0.001&0\\
&0.2&0.001&0&0\\
$n^{-1/3}$&0.3&0.003&0.009&0.011\\
&0.4&0.023&0.126&0.273\\
&0.5&0.062&0.519&0.848\\\hline
&0.1&0&0&0\\
&0.2&0.001&0&0\\
$n^{-1/2}$&0.3&0.001&0&0\\
&0.4&0.005&0.007&0.023\\
&0.5&0.023&0.089&0.308
\end{tabular}

\caption{Power for Kernel CvM Test under Design 2}
\label{kern_cvm_power_d2}
\end{table}

\begin{table}
\centering
\begin{tabular}{l|l|l|l|l}
$h_n$&$\theta_1-\overline\theta_1$&$n=100$&$n=500$&$n=1000$\\\hline
&0.1&0.001&0.001&0.001\\
&0.2&0.009&0.029&0.049\\
$n^{-1/5}$&0.3&0.044&0.185&0.386\\
&0.4&0.082&0.524&0.867\\
&0.5&0.18&0.879&0.997\\\hline
&0.1&0.007&0.015&0.014\\
&0.2&0.015&0.067&0.129\\
$n^{-1/3}$&0.3&0.029&0.18&0.454\\
&0.4&0.087&0.525&0.856\\
&0.5&0.167&0.825&0.98\\\hline
&0.1&0&0.014&0.006\\
&0.2&0&0.025&0.032\\
$n^{-1/2}$&0.3&0&0.057&0.123\\
&0.4&0&0.163&0.286\\
&0.5&0&0.321&0.604
\end{tabular}

\caption{Power for Kernel KS Test under Design 2}
\label{kern_ks_power_d2}
\end{table}

\begin{table}
\centering
\begin{tabular}{l|l|l|l}
$\theta_1-\overline\theta_1$&$n=100$&$n=500$&$n=1000$\\\hline
0.1&0.005&0&0.001\\
0.2&0.031&0.046&0.058\\
0.3&0.131&0.454&0.743\\
0.4&0.359&0.914&0.997\\
0.5&0.619&0.999&1
\end{tabular}

\caption{Power for Unweighted Instrument CvM Test under Design 3}
\label{uw_cvm_power_d3}
\end{table}

\begin{table}
\centering
\begin{tabular}{l|l|l|l}
$\theta_1-\overline\theta_1$&$n=100$&$n=500$&$n=1000$\\\hline
0.1&0.006&0.015&0.013\\
0.2&0.027&0.231&0.402\\
0.3&0.117&0.737&0.959\\
0.4&0.34&0.982&1\\
0.5&0.568&1&1
\end{tabular}

\caption{Power for Unweighted Instrument KS Test under Design 3}
\label{uw_ks_power_d3}
\end{table}

\begin{table}
\centering
\begin{tabular}{l|l|l|l|l}
$\sigma_n^2$&$\theta_1-\overline\theta_1$&$n=100$&$n=500$&$n=1000$\\\hline
&0.1&0.006&0&0.001\\
&0.2&0.037&0.079&0.136\\
$\frac{1}{4}n^{-1/5}$&0.3&0.133&0.515&0.837\\
&0.4&0.341&0.941&1\\
&0.5&0.636&1&1\\\hline
&0.1&0.006&0.003&0.001\\
&0.2&0.029&0.065&0.173\\
$\frac{1}{4}n^{-1/3}$&0.3&0.143&0.514&0.872\\
&0.4&0.375&0.961&1\\
&0.5&0.642&1&1\\\hline
&0.1&0.006&0.003&0\\
&0.2&0.043&0.059&0.101\\
$\frac{1}{4}n^{-1/2}$&0.3&0.161&0.52&0.845\\
&0.4&0.335&0.935&0.999\\
&0.5&0.63&0.999&1
\end{tabular}

\caption{Power for Weighted Instrument CvM Test under Design 3}
\label{w_cvm_power_d3}
\end{table}

\begin{table}
\centering
\begin{tabular}{l|l|l|l|l}
$t_n$&$\theta_1-\overline\theta_1$&$n=100$&$n=500$&$n=1000$\\\hline
&0.1&0.034&0.064&0.12\\
&0.2&0.093&0.466&0.704\\
$n^{-1/5}$&0.3&0.272&0.869&0.99\\
&0.4&0.501&0.994&1\\
&0.5&0.767&1&1\\\hline
&0.1&0.039&0.104&0.116\\
&0.2&0.112&0.429&0.64\\
$n^{-1/3}$&0.3&0.257&0.838&0.979\\
&0.4&0.463&0.994&1\\
&0.5&0.717&1&1\\\hline
&0.1&0.03&0.083&0.087\\
&0.2&0.121&0.325&0.523\\
$n^{-1/2}$&0.3&0.24&0.762&0.967\\
&0.4&0.397&0.984&1\\
&0.5&0.669&1&1
\end{tabular}

\caption{Power for Weighted Instrument KS Test under Design 3 (from \citet{armstrong_multiscale_2016})}
\label{w_ks_power_d3}
\end{table}

\begin{table}
\centering
\begin{tabular}{l|l|l|l|l}
$h_n$&$\theta_1-\overline\theta_1$&$n=100$&$n=500$&$n=1000$\\\hline
&0.1&0.013&0.017&0.018\\
&0.2&0.05&0.229&0.446\\
$n^{-1/5}$&0.3&0.187&0.757&0.965\\
&0.4&0.411&0.98&1\\
&0.5&0.698&1&1\\\hline
&0.1&0.007&0.012&0.01\\
&0.2&0.044&0.167&0.323\\
$n^{-1/3}$&0.3&0.173&0.676&0.932\\
&0.4&0.377&0.986&1\\
&0.5&0.657&1&1\\\hline
&0.1&0.002&0.001&0\\
&0.2&0.029&0.03&0.049\\
$n^{-1/2}$&0.3&0.082&0.326&0.654\\
&0.4&0.21&0.866&0.991\\
&0.5&0.47&0.996&1
\end{tabular}

\caption{Power for Kernel CvM Test under Design 3}
\label{kern_cvm_power_d3}
\end{table}

\begin{table}
\centering
\begin{tabular}{l|l|l|l|l}
$h_n$&$\theta_1-\overline\theta_1$&$n=100$&$n=500$&$n=1000$\\\hline
&0.1&0.043&0.087&0.161\\
&0.2&0.099&0.487&0.722\\
$n^{-1/5}$&0.3&0.261&0.876&0.99\\
&0.4&0.48&0.995&1\\
&0.5&0.746&1&1\\\hline
&0.1&0.037&0.086&0.122\\
&0.2&0.079&0.297&0.528\\
$n^{-1/3}$&0.3&0.164&0.646&0.912\\
&0.4&0.296&0.937&0.999\\
&0.5&0.507&0.996&1\\\hline
&0.1&0&0.035&0.026\\
&0.2&0&0.087&0.118\\
$n^{-1/2}$&0.3&0&0.195&0.385\\
&0.4&0&0.427&0.703\\
&0.5&0&0.716&0.952
\end{tabular}

\caption{Power for Kernel KS Test under Design 3}
\label{kern_ks_power_d3}
\end{table}

\clearpage

\begin{LARGE}
\begin{center}
Supplement to ``On the Choice of Test Statistic for Conditional Moment Inequalities''
\end{center}
\end{LARGE}

\begin{large}
\begin{center}
Timothy B. Armstrong  \\
Yale University
\end{center}
\end{large}

\begin{large}
\begin{center}
\today
\end{center}
\end{large}

\bigskip

This supplementary appendix contains proofs of the results in the main text as well as auxiliary results.
Section \ref{aux_append} contains auxiliary results used in the rest of this appendix.  These results are restatements or simple extensions of well known results on uniform convergence, and do not constitute part of the main novel contribution of the paper.
Section \ref{cval_append} of this appendix derives critical values for CvM
statistics with variance weights.  Section \ref{proof_append} contains
proofs of the results in the body of the paper.

\section{Auxiliary Results}\label{aux_append}

We state some results on uniform convergence that will be used in the proofs of the main results.  The results in this section are essentially restatements of results used in \citet{armstrong_weighted_2014}, which are in turn minor extensions of results in \citet{pollard_convergence_1984}.
Throughout this section, we consider iid observations $Z_1,\ldots,Z_n$ and a sequence of classes of functions $\mathcal{F}_n$ on the sample space.
Let $\sigma(f)^2=Ef(Z_i)^2-(Ef(Z_i))^2$ and
let $\hat\sigma(f)^2=E_nf(Z_i)^2-(E_nf(Z_i))^2$.

\begin{lemma}\label{unif_conv_lemma}
Suppose that $|f(Z_i)|\le \overline f$ a.s. and that
\begin{align*}
\sup_{n\in\mathbb{N}}\sup_{Q} N(\varepsilon,\mathcal{F}_n,L_1(Q))
  \le A \varepsilon^{-W}
\end{align*}
for some $A$ and $W$, where $N$ is the covering number defined in \citet{pollard_convergence_1984} and the supremum over $Q$ is over all probability measures.
Let $\sigma_n$ be a sequence of constants with $\sigma_n\sqrt{n/\log n}\to\infty$.  Then, for some constant $C$,
\begin{align*}
\frac{\sqrt{n}}{\sqrt{\log n}}\sup_{f\in\mathcal{F}_n}\left|\frac{(E_n-E)f(Z_i)}{\sigma(f)\vee \sigma_n}
  \right|
\le C
\end{align*}
with probability approaching one and
\begin{align*}
\sup_{f\in\mathcal{F}_n}\left|\frac{(E_n-E)f(Z_i)}{\sigma(f)^2\vee \sigma_n^2}
  \right|
\stackrel{p}{\to} 0.
\end{align*}
\end{lemma}
\begin{proof}
The first display follows by applying Lemma A.1 in \citet{armstrong_weighted_2014} to the sequence of classes of functions $\{f-E_Pf(Z_i)|f\in\mathcal{F}_n\}$, which satisfies the conditions of that lemma by Lemma A.5 in \citet{armstrong_weighted_2014}.  The second display follows from the first display since
\begin{align*}
\sup_{f\in\mathcal{F}_n}\left|\frac{(E_n-E)f(Z_i)}{\sigma(f)^2\vee \sigma_n^2}
  \right|
\le\frac{1}{\sigma_n} \sup_{f\in\mathcal{F}_n}\left|\frac{(E_n-E)f(Z_i)}{\sigma(f)\vee \sigma_n}
  \right|
=\frac{\sqrt{\log n}}{\sigma_n \sqrt{n}}\frac{\sqrt{n}}{\sqrt{\log n}} \sup_{f\in\mathcal{F}_n}\left|\frac{(E_n-E)f(Z_i)}{\sigma(f)\vee \sigma_n}
  \right|
\end{align*}
and $\sqrt{\log n}/(\sigma_n\sqrt{n})\to 0$.
\end{proof}

\begin{lemma}\label{sigma_unif_conv_lemma}
Under the conditions of Lemma \ref{unif_conv_lemma},
\begin{align*}
\sup_{f\in\mathcal{F}_n} \left|\frac{\hat\sigma(f)\vee \sigma_n}{\sigma(f)\vee \sigma_n}-1\right|\stackrel{p}{\to} 0.
\end{align*}
\end{lemma}
\begin{proof}
By continuity of $t\mapsto\sqrt{t}$ at $1$, it suffices to prove that
$\sup_{f\in\mathcal{F}_n} \left|\frac{\hat\sigma(f)^2\vee \sigma_n^2}{\sigma(f)^2\vee \sigma_n^2}-1\right|\stackrel{p}{\to} 0$.
We have
\begin{align*}
&\sup_{f\in\mathcal{F}_n} \left|\frac{\hat\sigma(f)^2\vee \sigma_n^2}{\sigma(f)^2\vee \sigma_n^2}-1\right|
=\sup_{f\in\mathcal{F}_n} \left|\frac{\hat\sigma(f)^2\vee \sigma_n^2-\sigma(f)^2\vee \sigma_n^2}{\sigma(f)^2\vee \sigma_n^2}\right|
\le \sup_{f\in\mathcal{F}_n} \left|\frac{\hat\sigma(f)^2-\sigma(f)^2}{\sigma(f)^2\vee \sigma_n^2}\right|.
\end{align*}
Note that
\begin{align}\label{samp_var_decom_eq}
\hat\sigma(f)^2-\sigma(f)^2
=(E_n-E)[f(Z_i)-Ef(Z_i)]^2-[(E_n-E)f(Z_i)]^2.
\end{align}
Since
$\sigma[(f-Ef(Z_i))^2]^2
\le E[f(Z_i)-Ef(Z_i)]^4
\le 4\overline f^2\sigma(f)^2$,
we have
\begin{align*}
\sup_{f\in\mathcal{F}_n}
\frac{\left|(E_n-E)[f(Z_i)-Ef(Z_i)]^2\right|}{\sigma(f)^2\vee\sigma_n^2}
\le \sup_{f\in\mathcal{F}_n}
\frac{\left|(E_n-E)[f(Z_i)-Ef(Z_i)]^2\right|}{\sigma[(f-Ef(Z_i))^2]^2\vee\sigma_n^2}
\cdot (4\overline f^2)\vee 1
\end{align*}
which converges in probability to zero by Lemma \ref{unif_conv_lemma} (using Lemma A.5 in \citealp{armstrong_weighted_2014} to verify that the sequence of classes of functions $\{[f-Ef(Z_i)]^2|f\in\mathcal{F}_n\}$ satisfies the conditions of the lemma).  Since
\begin{align*}
\frac{[(E_n-E)f(Z_i)]^2}{\sigma(f)^2\vee \sigma_n^2}\stackrel{p}{\to} 0
\end{align*}
by Lemma \ref{unif_conv_lemma}, the result now follows from this and the triangle inequality applied to (\ref{samp_var_decom_eq}).
\end{proof}

\begin{lemma}\label{lp_bound_lemma}
Suppose that $|f(Z_i)|\le \overline f$ and that
$\sigma_n\sqrt{n}\ge 1$.
Then
\begin{align*}
E\left|\frac{\sqrt{n}(E_n-E)f(Z_i)}{\sigma(f)\vee \sigma_n}\right|^p
  \le C_{p,\overline f}
\end{align*}
for a constant $C_{p,\overline f}$ that depends only on $p$ and $\overline f$.
\end{lemma}
\begin{proof}
By Bernstein's inequality,
\begin{align*}
&P\left(\left|\frac{\sqrt{n}(E_n-E)f(Z_i)}{\sigma(f)\vee \sigma_n}\right| > t\right)
  \le \exp\left(-\frac{1}{2}\frac{n[\sigma(f)\vee \sigma_n]^2 t^2}
        {n \sigma^2(f)+\frac{1}{3}\cdot 2\overline f\cdot \sqrt{n}[\sigma(f)\vee \sigma_n] t}\right)  \\
&\le \exp\left(-\frac{1}{2}\frac{ t^2}
        {1+\frac{1}{3}\cdot 2\overline f\cdot \frac{t}{\sqrt{n}[\sigma(f)\vee \sigma_n]}}\right)
\le \exp\left(-\frac{1}{2}\frac{ t^2}
        {1+\frac{1}{3}\cdot 2\overline f\cdot t}\right)
\le \exp\left(-\frac{1}{2}\frac{ t^2}
        {1+\frac{1}{3}\cdot 2\overline f\cdot t}\right).
\end{align*}
For $t\ge 1$, this is bounded by
$\exp\left(-\frac{ t}
        {2+\frac{2}{3}\cdot 2\overline f}\right)$.
Thus,
\begin{align*}
&E\left|\frac{\sqrt{n}(E_n-E)f(Z_i)}{\sigma(f)\vee \sigma_n}\right|^p
=\int_{t=0}^{\infty} P\left(\left|\frac{\sqrt{n}(E_n-E)f(Z_i)}{\sigma(f)\vee \sigma_n}\right|^p > t\right) \, dt  \\
&\le 1+ \int_{t=1}^{\infty} \exp\left(-\frac{ t^{1/p}}
        {2+\frac{2}{3}\cdot 2\overline f}\right)\, dt
\end{align*}
which is finite and depends only on $p$ and $\overline f$ as claimed.
\end{proof}

\section{Critical Values for CvM Statistics with Variance Weights}\label{cval_append}

For bounded choices of
$\omega$ (which corresponds to $\sigma_n$ bounded away from zero when a
truncated variance weighting is used), \citet{kim_kyoo_il_set_2008} and
\citet{andrews_inference_2013} derive a $\sqrt{n}$ rate of convergence to
an asymptotic distribution that may be degenerate.
\citet{armstrong_weighted_2014} shows that letting $\sigma_n$ go to zero
generally decreases the rate of convergence to $\sqrt{n/\log n}$ for the
KS statistic $T_{n,\infty,\omega}$.  In contrast to the KS case, CvM statistics do not behave much differently if the variance is allowed to go to zero, although some additional arguments are needed to show this.

To deal with the behavior of the CvM statistic for small variances, I place the following condition on the measure over which the sample means are integrated.

\begin{assumption}\label{small_var_mu_assump}
$\mu(\{g|\sigma_j(\theta,g)\le \delta\})\to 0$ as $\delta\to 0$ for all
$j$.
\end{assumption}

This condition will hold for the choices of $\mathcal{G}$ and $\mu$
used in the body of the paper, and also allow for more general choices
of $\mathcal{G}$ and $\mu$.
I also make the following assumption on
the complexity of the class of functions $\mathcal{G}$, which is also
satisfied by the class used in the paper.

\begin{assumption}\label{covering_assump}
For some constants $A$ and $\varepsilon$,
the covering number $N(\varepsilon,\mathcal{G},L_1(Q))$ defined in
\citet{pollard_convergence_1984} satisfies
\begin{align*}
\sup_{Q} N(\varepsilon,\mathcal{G},L_1(Q)) \le A\varepsilon^{-W},
\end{align*}
whre the supremum is over all probability measures.
\end{assumption}

The following condition imposes a bounded distribution of the function $m$.

\begin{assumption}\label{bdd_assump}
For some nonrandom constant $\overline Y$,
$|m_j(W_i,\theta)|\le \overline Y$ for each $j$ with probability one.
\end{assumption}

\begin{theorem}
Suppose that $\sigma_n\sqrt{n/\log n}\to\infty$ and that Assumptions \ref{small_var_mu_assump}, \ref{covering_assump} and \ref{bdd_assump} hold.  Then, for $\theta\in\Theta_0$,
\begin{align*}
&n^{1/2}T_{n,p,(\hat\sigma\vee\sigma_n)^{-1},\mu}(\theta)
  \le\left[\int \sum_{j=1}^{d_Y}
  \left|\frac{\sqrt{n}(E_n-E)m_j(W_i,\theta)g(X_i)}{\hat\sigma_j(\theta,g)\vee \sigma_n}\right|_{-}^p \,
    d\mu(g)\right]^{1/p}  \\
&\stackrel{d}{\to} \left[\int \sum_{j=1}^{d_Y}
|\mathbb{G}_j(g,\theta)/\sigma_j(\theta,g)|_{-}^p \,
    d\mu(g)\right]^{1/p}
\end{align*}
where $\mathbb{G}(g,\theta)$ is a vector of Gaussian processes with covariance function
\begin{align*}
\rho(g,\tilde g)
=E[m(W_i,\theta)g(X_i)-Em(W_i,\theta)g(X_i)]
  [m(W_i,\theta)\tilde g(X_i)-Em(W_i,\theta)\tilde g(X_i)]'.
\end{align*}

\end{theorem}
\begin{proof}
The result with the integral truncated over
$\{\sigma_j(\theta,g)\le \delta| \text{all $j$}\}$ follows immediately from standard arguments using functional central limit theorems.
This, along with Lemma \ref{int_bound_lemma} below gives, letting $Z_n(\delta)$ be the integral truncated at $\{\sigma_j(\theta,g)\le \delta| \text{all $j$}\}$ and $Z(\delta)$ be the limiting variable with this truncation,
\begin{align*}
P(Z_n(\delta)-\varepsilon\le t)-\varepsilon
\le P(n^{1/2}T_{n,p,\omega,\mu}(\theta)\le t)
\le P(Z_n(\delta)\le t)
\end{align*}
for large enough $n$ for any $\varepsilon>0$.
The $\liminf$ of the left hand size is greater than
$P(Z(\delta)\le t-2\varepsilon)-2\varepsilon$,
and the $\limsup$ of the right hand side is less than
$P(Z(\delta)\le t+\varepsilon)+\varepsilon$.  We can bound
$P(Z(\delta)\le t-2\varepsilon)-2\varepsilon$ from below by $P(Z\le
t-2\varepsilon)-2\varepsilon$, and we can bound $P(Z(\delta)\le
t+\varepsilon)+\varepsilon$ from above by
$P(Z\le t+2\varepsilon)+2\varepsilon$ by making $\delta$ small enough
by a version of Lemma \ref{int_bound_lemma} for the limiting process.
Since $\varepsilon$ was arbitrary, this gives the result.
\end{proof}

The proof of the theorem above uses the following auxiliary lemma,
which shows that functions $g$ with low enough variance have little
effect on the integral asymptotically.

\begin{lemma}\label{int_bound_lemma}
Fix $j$ and suppose that Assumptions \ref{small_var_mu_assump},
\ref{covering_assump} and \ref{bdd_assump} hold, and that the null
hypothesis holds under $\theta$.  Then, for every
$\varepsilon>0$, there exists a $\delta>0$ such that
\begin{align*}
P\left( \sqrt{n}\left[\int_{\sigma_j(\theta,g)\le\delta} 
    |E_nm_j(W_i,\theta)g(X_i)
      /(\hat\sigma_j(\theta,g)\vee\sigma_n)|_{-}^p \, d\mu(g)\right]^{1/p}
    >\varepsilon\right)
\le \varepsilon.
\end{align*}
\end{lemma}
\begin{proof}
We have
\begin{align*}
&E\int_{\sigma_j(\theta,g)\le\delta} 
    |\sqrt{n}E_nm_j(W_i,\theta)g(X_i)
      /(\sigma_j(\theta,g)\vee\sigma_n)|_{-}^p \, d\mu(g)  \\
&=\int_{\sigma_j(\theta,g)\le\delta} 
    E|\sqrt{n}E_nm_j(W_i,\theta)g(X_i)
      /(\sigma_j(\theta,g)\vee\sigma_n)|_{-}^p \, d\mu(g)  \\
&\le \int_{\sigma_j(\theta,g)\le\delta} 
    E|\sqrt{n}(E_n-E)m_j(W_i,\theta)g(X_i)
      /(\sigma_j(\theta,g)\vee\sigma_n)|^p \, d\mu(g)
\le \mu\left(\left\{g|\sigma_j(\theta,g)\le\delta\right\}\right)
  \cdot C_{p,\overline Y}
\end{align*}
for $C_{p,\overline Y}$ given in Lemma \ref{lp_bound_lemma}.  Applying Markov's inequality and using Assumption \ref{small_var_mu_assump}, it follows that, for any $\varepsilon>0$, there exists a $\delta$ such that
\begin{align*}
P\left(\sqrt{n}\left[\int_{\sigma_j(\theta,g)\le\delta} 
    |E_nm_j(W_i,\theta)g(X_i)
      /(\sigma_j(\theta,g)\vee\sigma_n)|_{-}^p \, d\mu(g)\right]^{1/p} > \varepsilon/2\right)
\le \varepsilon/2.
\end{align*}
The result follows since
\begin{align*}
&\sqrt{n}\left[\int_{\sigma_j(\theta,g)\le\delta} 
    |E_nm_j(W_i,\theta)g(X_i)
      /(\hat\sigma_j(\theta,g)\vee\sigma_n)|_{-}^p \, d\mu(g)\right]^{1/p}  \\
&\le \sqrt{n}\left[\int_{\sigma_j(\theta,g)\le\delta} 
    |E_nm_j(W_i,\theta)g(X_i)
      /(\sigma_j(\theta,g)\vee\sigma_n)|_{-}^p \, d\mu(g)\right]^{1/p}
\cdot \sup_{g} (\sigma_j(\theta,g)\vee\sigma_n)/(\hat\sigma_j(\theta,g)\vee\sigma_n)
\end{align*}
and $\sup_{g} (\sigma_j(\theta,g)\vee\sigma_n)/(\hat\sigma_j(\theta,g)\vee\sigma_n)\le 2$ with probability approaching one by Lemma \ref{sigma_unif_conv_lemma}.
\end{proof}

\section{Proofs}\label{proof_append}

This section contains proofs of the results in the body of the paper.
The proofs use a number of auxiliary lemmas, which are stated and
proved first.  In the following, $\theta_n$ is always assumed to be a sequence converging to $\theta_0$.

\begin{lemma}\label{kern_unif_lemma}
Under the assumptions of Theorem \ref{kern_lim_thm}, there exists a constant $C$ such that
\begin{align*}
\sup_{x\in\mathbb{R}^{d_X}}\frac{\sqrt{n}}{\sqrt{h^{d_X}\log n}}\left|(E_n-E)m(W_i,\theta_n)k((X_i-x)/h)\right|\le C
\end{align*}
and
\begin{align*}
\sup_{x\in\mathbb{R}^{d_X}}\frac{\sqrt{n}}{\sqrt{h^{d_X}\log n}}\left|(E_n-E)k((X_i-x)/h)\right|
  \le C
\end{align*}
with probability approaching one.  In addition,
\begin{align*}
\sup_{\{x|\omega_j(\theta_n,x)>0 \text{ some $j$}\}}
  \left|\frac{E_nk((X_i-h)/h)}{Ek((X_i-h)/h)}-1\right|\stackrel{p}{\to} 0.
\end{align*}
\end{lemma}
\begin{proof}
The first two displays follow from Lemma \ref{unif_conv_lemma} after noting that
\begin{align*}
var(m(W_i,\theta_n)k((X_i-x)/h))
\le \overline Y^2\overline k^2\overline f_X B^{d_X} h^{d_X}
\end{align*}
where $\overline k$ and $\overline f_X$ are bounds for $k$ and $f_X$, and $B$ is such that $k(u)=0$ whenever $\max_{1\le j\le d_X}|u_j|>B/2$, and similarly for
$var(k((X_i-x)/h))$, and that $\sqrt{h^{d_X}}\sqrt{n}/\sqrt{\log n}\to \infty$ under these assumptions.

For the last display, note that, for $x$ such that $\omega_j(\theta_n,x)>0$ for some $j$, $Ek((X_i-x)/h)\ge \underline f_X h^{d_X}\int k(u)\, du$ for large enough $n$, where $\underline f_X$ is a lower bound for the density of $X_i$ (which can be taken to be $\varepsilon$ in Assumption \ref{kern_dens_assump}).  Thus,
\begin{align*}
&\sup_{\{x|\omega_j(\theta_n,x)>0 \text{ some $j$}\}}
  \left|\frac{E_nk((X_i-h)/h)}{Ek((X_i-h)/h)}-1\right|
\le \sup_{x\in\mathbb{R}^{d_X}}
  \left|\frac{(E_n-E)k((X_i-h)/h)}{\underline f_X h^{d_X}\int k(u)\, du}\right|  \\
&=\sup_{x\in\mathbb{R}^{d_X}}
  \frac{\sqrt{n}}{\sqrt{h^{d_X}\log n}}
     \left|(E_n-E)k((X_i-h)/h)\right|
\cdot \frac{\sqrt{h^{d_X}\log n}}{\sqrt{n}\underline f_X h^{d_X}\int k(u)\, du}.
\end{align*}
The result then follows from the second display, since
$\frac{\sqrt{\log n}}{\sqrt{n h^{d_X}}}\to 0$.
\end{proof}

Let
\begin{align*}
\tilde T_{n,p,(\hat\sigma\vee \sigma_n)^{-1},\mu}(\theta)
  =\left[\int_{h>0}\int_x
\sum_{j=1}^{d_Y}\left|\frac{E_nm(W_i,\theta)k((X_i-x)/h)}{\sigma_j(\theta,x,h)\vee \sigma_n}\right|_{-}^pf_{\mu}(x,h)\, dx\, dh\right]^{1/p}
\end{align*}
and let
\begin{align*}
\tilde T_{n,p,\text{kern}}(\theta)
  =\left[\int_x
\sum_{j=1}^{d_Y}\left|\frac{E_nm(W_i,\theta)k((X_i-x)/h)}{Ek((X_i-x)/h)}\right|_{-}^p\omega_j(\theta,x)\, dx\, dh\right]^{1/p}.
\end{align*}
The notation $\sigma_j(\theta,\tilde x,h)$ is used to denote $\sigma_j(\theta,g)$ where $g(x)=k((x-\tilde x)/h)$.

\begin{lemma}\label{equiv_lemma_1}
Under Assumptions
\ref{g_kernel_assump},
\ref{mu_assump},
\ref{smoothness_assump_multi} and
\ref{bdd_y_assump_local},
\begin{align*}
\sqrt{n}T_{n,p,(\hat\sigma\vee \sigma_n)^{-1},\mu}(\theta_n)
= \sqrt{n}\tilde T_{n,p,(\hat\sigma\vee \sigma_n)^{-1},\mu}(\theta_n)(1+o_P(1))
\end{align*}
for any sequence $\theta_n\to\theta_0$.
If Assumption \ref{kern_dens_assump} holds as well, then
\begin{align*}
(nh^{d_X})^{1/2}T_{n,p,\text{kern}}(\theta_n)
  =(nh^{d_X})^{1/2}\tilde T_{n,p,\text{kern}}(\theta_n)(1+o_P(1))
\end{align*}
for any sequence $\theta_n\to\theta_0$.
\end{lemma}
\begin{proof}
We have
\begin{align*}
|\sqrt{n}T_{n,p,(\hat\sigma\vee \sigma_n)^{-1},\mu}(\theta_n)-\sqrt{n}\tilde T_{n,p,(\hat\sigma\vee \sigma_n)^{-1},\mu}(\theta_n)|
\le \sqrt{n}\tilde T_{n,p,(\hat\sigma\vee \sigma_n)^{-1},\mu}(\theta)
  \cdot \sup_{x,j} \left|\frac{\sigma_j(\theta_n,x,h)\vee \sigma_n}{\hat \sigma_j(\theta_n,x,h)\vee \sigma_n}-1\right|.
\end{align*}
Thus, the first display follows from Lemma \ref{sigma_unif_conv_lemma}.

Similarly, for the second display,
\begin{align*}
&|(nh^{d_X})^{1/2}T_{n,p,\text{kern}}(\theta_n)-(nh^{d_X})^{1/2}\tilde T_{n,p,\text{kern}}(\theta_n)|  \\
&\le (nh^{d_X})^{1/2}\tilde T_{n,p,\text{kern}}(\theta_n)
  \cdot \sup_{\{x|\omega_j(\theta,x)>0 \text{ some $j$}\}} 
     \left|\frac{Ek((X_i-x)/h)}{E_nk((X_i-x)/h)}-1\right|,
\end{align*}
and the result follows from Lemma \ref{kern_unif_lemma}.
\end{proof}

Let
\begin{align*}
\tilde{\tilde T}_{n,p,(\hat\sigma\vee \sigma_n)^{-1},\mu}(\theta)
  =\left[\int_{h>0}\int_x
\sum_{j=1}^{d_Y}\left|\frac{Em(W_i,\theta)k((X_i-x)/h)}{\sigma_j(\theta,x,h)\vee \sigma_n}\right|_{-}^pf_{\mu}(x,h)\, dx\, dh\right]^{1/p}
\end{align*}
and let
\begin{align*}
\tilde{\tilde T}_{n,p,\text{kern}}(\theta)
  =\left[\int_x
\sum_{j=1}^{d_Y}\left|\frac{Em(W_i,\theta)k((X_i-x)/h)}{Ek((X_i-x)/h)}\right|_{-}^p\omega_j(\theta,x)\, dx\, dh\right]^{1/p}.
\end{align*}
Also define
\begin{align*}
\tilde{\tilde T}_{n,p,1,\mu}(\theta)
  =\left[\int_{h>0}\int_x
\sum_{j=1}^{d_Y}\left|Em(W_i,\theta)k((X_i-x)/h)\right|_{-}^pf_{\mu}(x,h)\, dx\, dh\right]^{1/p}.
\end{align*}

\begin{lemma}\label{inst_equiv_lemma_2}
Under Assumptions
\ref{g_kernel_assump},
\ref{mu_assump},
\ref{smoothness_assump_multi} and
\ref{bdd_y_assump_local},
\begin{align*}
\sqrt{n}\tilde T_{n,p,(\hat\sigma\vee \sigma_n)^{-1},\mu}(\theta_n)
=\sqrt{n}\tilde{\tilde T}_{n,p,(\hat\sigma\vee \sigma_n)^{-1},\mu}(\theta_n)+o_{P}(1).
\end{align*}
and
\begin{align*}
\sqrt{n} T_{n,p,1,\mu}(\theta_n)
=\sqrt{n}\tilde{\tilde T}_{n,p,1,\mu}(\theta_n)+o_{P}(1).
\end{align*}
\end{lemma}
\begin{proof}
Let $\tilde \sigma_n\to 0$ be such that
$\tilde\sigma_n\sqrt{n/\log n}\to\infty$ and
$\tilde\sigma_n/\sigma_n\to 0$ (i.e. $\tilde \sigma_n$ is chosen to be much smaller than $\sigma_n$, but such that the assumptions still hold for $\tilde \sigma_n$).  Note that
\begin{align*}
&\sqrt{n}|\tilde{\tilde T}_{n,p,(\hat\sigma\vee \sigma_n)^{-1},\mu}(\theta_n)
-\tilde T_{n,p,(\hat\sigma\vee \sigma_n)^{-1},\mu}(\theta_n)|  \\
&\le \left[\int\int_{(x,h)\in \hat{\mathcal{G}}}
\sum_{j=1}^{d_Y}\left|\sqrt{n}\frac{(E_n-E)m(W_i,\theta_n)k((X_i-x)/h)}{\sigma_j(\theta,x,h)\vee \sigma_n}\right|^pf_{\mu}(x,h)\, dx\, dh\right]^{1/p}
\end{align*}
where $\hat{\mathcal{G}}=\{(x,h)|Em(W_i,\theta_n)k((X_i-x)/h)<0 \text{ or }
  E_n(W_i,\theta_n)k((X_i-x)/h)<0\}$.

For any $\varepsilon>0$, there exists an $\eta>0$ such that, for $h>\varepsilon$ and large enough $n$,
\begin{align*}
Em_j(W_i,\theta_n)k((X_i-x)/h)\ge \eta Ek((X_i-x)/h)
\ge \eta \cdot var[m_j(W_i,\theta_n)k((X_i-x)/h)]\cdot \frac{1}{\overline k\overline Y^2}
\end{align*}
where the second inequality follows since
\begin{align*}
var[m_j(W_i,\theta_n)k((X_i-x)/h)]
\le \overline Y^2E[k((X_i-x)/h)^2]
\le \overline Y^2\overline kEk((X_i-x)/h).
\end{align*}
Thus, for large enough $n$ we will have
\begin{align*}
&E_nm_j(W_i,\theta_n)k((X_i-x)/h)  \\
&\ge (E_n-E)m_j(W_i,\theta_n)k((X_i-x)/h)
  +var[m_j(W_i,\theta_n)k((X_i-x)/h)]\cdot \frac{\eta}{\overline k\overline Y^2},
\end{align*}
and the last line is positive for all $(x,h)$ with $\sigma_j(\theta_n,x,h)\ge \tilde\sigma_n$ with probability approaching one by Lemma \ref{unif_conv_lemma}.

From this and the fact that $Em(W_i,\theta_n)k((X_i-x)/h)\ge 0$ for all $h>\varepsilon$ for large enough $n$, it follows that
$\hat{\mathcal{G}}\subseteq \{(x,h)|h\le \varepsilon \text{ or } \sigma_j(\theta_n,x,h)<\tilde\sigma_n\}$
with probability approaching one.  Note that
\begin{align*}
&E\int\int_{\{(x,h)|h\le \varepsilon\}}
\sum_{j=1}^{d_Y}\left|\frac{\sqrt{n}(E_n-E)m(W_i,\theta_n)k((X_i-x)/h)}{\sigma_j(\theta,x,h)\vee \sigma_n}\right|^pf_{\mu}(x,h)\, dx\, dh  \\
&=\int\int_{\{(x,h)|h\le \varepsilon\}}
\sum_{j=1}^{d_Y}E\left|\frac{\sqrt{n}(E_n-E)m(W_i,\theta_n)k((X_i-x)/h)}{\sigma_j(\theta,x,h)\vee \sigma_n}\right|^pf_{\mu}(x,h)\, dx\, dh
\end{align*}
by Fubini's theorem, and this can be made arbitrarily small by making $\varepsilon$ small by Lemma \ref{lp_bound_lemma} and Assumption \ref{mu_assump}.
Similarly,
\begin{align*}
&E\int\int_{\{(x,h)|\sigma_j(\theta_n,x,h)< \tilde\sigma_n \text{ some $j$}\}}
\sum_{j=1}^{d_Y}\left|\frac{\sqrt{n}(E_n-E)m(W_i,\theta_n)k((X_i-x)/h)}{\sigma_j(\theta,x,h)\vee \sigma_n}\right|^pf_{\mu}(x,h)\, dx\, dh  \\
&\le \mu(\mathbb{R}^{d_X}\times [0,\infty))
\cdot \sup_{\{(x,h,j)|\sigma_j(\theta_n,x,h)< \tilde\sigma_n\}} E\left|\frac{\sqrt{n}(E_n-E)m(W_i,\theta_n)k((X_i-x)/h)}{\sigma_j(\theta,x,h)\vee \sigma_n}\right|^p  \\
&= \mu(\mathbb{R}^{d_X}\times [0,\infty))
\cdot \sup_{\{(x,h,j)|\sigma_j(\theta_n,x,h)< \tilde\sigma_n\}} E\left|\frac{\sqrt{n}(E_n-E)m(W_i,\theta_n)k((X_i-x)/h)}{\sigma_j(\theta,x,h)\vee \tilde \sigma_n}\right|^p\frac{\tilde\sigma_n}{\sigma_n},
\end{align*}
which converges to zero by Lemma \ref{lp_bound_lemma}.
Using this and Markov's inequality, it follows that
$\sqrt{n}|\tilde{\tilde T}_{n,p,(\hat\sigma\vee \sigma_n)^{-1},\mu}(\theta)
-\tilde T_{n,p,(\hat\sigma\vee \sigma_n)^{-1},\mu}(\theta)|$ can be made arbitrarily small with probability approaching one by making $\varepsilon$ small.  This gives the first display of the lemma.

The second display follows by the same argument with $\sigma_n$ set to the supremum of $\sigma_j(\theta,x,h)$ over $x$, $h$ on the support of $\mu$, $\theta$ in a neighborhood of $\theta_0$ and all $j$.
\end{proof}

\begin{lemma}\label{kern_equiv_lemma_2}
Under Assumptions
\ref{g_kernel_assump},
\ref{mu_assump},
\ref{smoothness_assump_multi},
\ref{bdd_y_assump_local} and
\ref{kern_dens_assump},
\begin{align*}
(nh^{d_X})^{1/2}\tilde T_{n,p,\text{kern}}(\theta_n)
  =(nh^{d_X})^{1/2}\tilde{\tilde T}_{n,p,\text{kern}}(\theta_n)+o_P(1).
\end{align*}
\end{lemma}
\begin{proof}
For any $\varepsilon>0$, there is an $\eta>0$ such that $Em_j(W_i,\theta_n)k((X_i-x)/h)>\eta Ek((X_i-x)/h)$
for all $x\in \bar{\mathcal{X}}(\varepsilon)$ where
$\bar{\mathcal{X}}(\varepsilon)$ is the set of $x$
with $\|x-x_k\|\ge \varepsilon$ for all $k=1,\ldots,\ell$
and $\omega_j(\theta_n,x)>0$ for some $j$.  Thus, arguing as in Lemma \ref{inst_equiv_lemma_2} and using Lemma \ref{kern_unif_lemma}, it follows that, with probability approaching one,
\begin{align*}
&(nh^{d_X})^{1/2}|\tilde T_{n,p,\text{kern}}(\theta_n)-\tilde{\tilde T}_{n,p,\text{kern}}(\theta_n)|  \\
&\le \left[\int_{x\not\in \bar{\mathcal{X}}(\varepsilon)} \sum_{j=1}^{d_Y} \left|\frac{\sqrt{nh^{d_X}}(E_n-E)m_j(W_i,\theta_n)k((X_i-x)/h)}{Ek((X_i-x)/h)}\right|^p\omega_j(\theta_n,x)\, dx\right]^{1/p}.
\end{align*}
Using Markov's inequality and Fubini's theorem along with the fact that
$\int_{x\not\in\bar{\mathcal{X}}(\varepsilon)} w_j(\theta_nx)\, dx$ can be made arbitrarily small by making $\varepsilon$ small, the result follows so long as
\begin{align*}
E\left|\frac{\sqrt{nh^{d_X}}(E_n-E)m_j(W_i,\theta_n)k((X_i-x)/h)}{Ek((X_i-x)/h)}\right|^p
\end{align*}
can be bounded uniformly over $x$ such that $\omega_j(\theta_n,x)>0$.  But this follows from Lemma \ref{lp_bound_lemma}, since, by Assumptions \ref{g_kernel_assump} and \ref{kern_dens_assump}, for some $\delta>0$,
$Ek((X_i-x)/h)\ge \delta h^{d_X}$ for all $x$ with $\omega_j(\theta_n,x)>0$.
\end{proof}

For the following lemma, recall that
$w_j(x_k)=(s_j^2(x_k,\theta_0)f_X(x_k)\int k(u)^2\, du)^{-1/2}$
and $s_j^2(x,\theta)=var(m(W_i,\theta)|X_i=x)$.

\begin{lemma}\label{sigma_lim_lemma}
Under Assumptions
\ref{g_kernel_assump},
\ref{mu_assump},
\ref{smoothness_assump_multi} and
\ref{bdd_y_assump_local},
for $k=1,\ldots,\ell$
\begin{align*}
\sup_{\|(x,h)-(x_k,0)\|\le \varepsilon_n}
\left|h^{-d_X/2}\sigma_j(\theta_n,x,h)-w_j(x_k)^{-1}\right|
\to 0.
\end{align*}
for any sequences $\varepsilon_n\to 0$ and $\theta_n\to \theta_0$.
\end{lemma}
\begin{proof}
By differentiability of the square root function at $w_j^{-2}(x_k)$, it suffices to show that
$\sup_{\|(x,h)-(x_k,0)\|\le \varepsilon_n}
\left|h^{-d_X}\sigma^2_j(\theta_n,x,h)-w^{-2}_j(x_k)\right|
\to 0$.
Note that
\begin{align*}
&h^{-d_X}\sigma_j^2(\theta_n,x,h)
=h^{-d_X}E[m(W_i,\theta_n)^2k((X_i-x)/h)^2]
  -h^{-d_X}\{E[m(W_i,\theta_n)k((X_i-x)/h)]\}^2  \\
&=h^{-d_X}\int s_j^2(\tilde x,\theta_n) k((\tilde x-x)/h)^2f_X(\tilde x)\, d\tilde x  \\
  &+h^{-d_X}\int E[m(W_i,\theta_n)|X_i=\tilde x]^2k((\tilde x-x)/h)^2f_X(\tilde x)\, d\tilde x  \\
  &-h^{-d_X}\left\{\int E[m(W_i,\theta_n)|X_i=\tilde x]k((\tilde x-x)/h)f_X(\tilde x)\, d\tilde x\right\}^2.
\end{align*}
By Assumption \ref{g_kernel_assump} and part (iii) of Assumption \ref{smoothness_assump_multi}, the second term is bounded by a constant times
$\sup_{\|(x,h)-(x_k,0)\|\le \varepsilon_n}E[m(W_i,\theta_n)|X_i=x]^2$, which converges to zero by
continuity of $E[m(W_i,\theta)|X_i=x]$ at $(\theta_0,x_k)$.
By Assumptions \ref{g_kernel_assump} and \ref{smoothness_assump_multi},
the third term is bounded by a constant times $h^{-d_X}\cdot h^{2 d_X}\le \varepsilon_n^{d_X}$ uniformly over $(x,h)$ with $\|(x,h)-(x_k,0)\|\le \varepsilon_n$.
Using a change of variables, the first term can be written as
$\int s_j^2(x+uh,\theta_n)k(u)^2f_X(x+uh)\, du$,
which converges to $w_j^{-2}(x_k)$ uniformly over $\|(x,h)-(x_k,0)\|\le \varepsilon_n$ by continuity of $s_j$ and $f_X$, and by Assumption \ref{g_kernel_assump}.
\end{proof}

\begin{lemma}\label{k_lim_lemma}
Suppose that Assumptions
\ref{g_kernel_assump},
\ref{mu_assump},
\ref{smoothness_assump_multi},
\ref{bdd_y_assump_local} and
\ref{kern_dens_assump}
hold, and that $\int k(u)\, du=1$.  Then
\begin{align*}
\sup_{\|x-x_k\|\le \varepsilon}
|h^{-d_X}Ek((X_i-x)/h)-f_X(x_k)|\to 0
\end{align*}
as $h\to 0$ and $\varepsilon\to 0$ for $k=1,\ldots, \ell$.
\end{lemma}
\begin{proof}
We have
\begin{align*}
h^{-d_X}Ek((X_i-x)/h)
=h^{-d_X}\int k((\tilde x-x)/h)f_X(\tilde x)\, d\tilde x
=\int k(u) f_X(x+uh)\, du,
\end{align*}
and $\int k(u)\, du=1$ and $f_X(x+uh)$ converges to $f_X(x_k)$ uniformly over $\|x-x_k\|\le \varepsilon$ and $u$ in the support of $k$ as
$\varepsilon\to 0$ and $h\to 0$.
\end{proof}

For notational convenience in the following lemmas, define, for $(j,k)$ with
 $j\in J(k)$,
\begin{align*}
\tilde \psi_{j,k}(x-x_k)
=\frac{\bar m_j(\theta_0,x)-\bar m_j(\theta_0,x_k)}{\|x-x_k\|^{\gamma(j,k)}}
\end{align*}
so that
\begin{align*}
\sup_{\|x-x_k\|<\delta}\left|\tilde \psi_{j,k}(x-x_k)
  -\psi_{j,k}\left(\frac{x-x_k}{\|x-x_k\|}\right)\right|
\to 0
\end{align*}
under Assumption \ref{smoothness_assump_multi}.

\begin{lemma}\label{drift_approx_bdd_lemma}
Under Assumptions
\ref{g_kernel_assump},
\ref{mu_assump},
\ref{smoothness_assump_multi} and
\ref{bdd_y_assump_local},
for any $a\in\mathbb{R}^{d_\theta}$,
\begin{align*}
&r^{-[d_X+p(d_X+\gamma)+1]/\gamma}
\int\int \sum_{j=1}^{d_Y}|E m_j(W_i,\theta_0+ra)k((X_i-\tilde x)/h)|_{-}^p
    f_\mu(\tilde x,h)\, d\tilde x \, dh  \\
&\stackrel{r\to 0}{\to}
\sum_{k=1}^{\mathcal{X}_0} \sum_{j\in \tilde J(k)}
  \lambda_{\text{bdd}}(a,j,k,p).
\end{align*}
\end{lemma}
\begin{proof}
For simplicity, assume that $\gamma(j,k)=\gamma$ for all $j,k$.  The general result follows from applying the same arguments to show that areas of $(x,h)$ near $(j,k)$ with $\gamma(j,k)<\gamma$ do not matter asymptotically.

For $C$ large enough, the integrand will be zero unless
$\max\{\|\tilde x-x_k\|,h\}<C r^{1/\gamma}$ for some $k$ with $j\in J(k)$.  Thus, it suffices to prove the lemma for, fixing $(j,k)$ with $j\in J(k)$,
\begin{align*}
&\int\int
|E m_j(W_i,\theta_0+ra)k((X_i-\tilde x)/h)|_{-}^p
    f_\mu(\tilde x,h)\, d\tilde x \, dh  \\
&=\int\int
\left|\int \bar m_j(\theta_0+ra,x)k((x-\tilde x)/h)f_X(x)\, dx\right|_{-}^p
    f_\mu(\tilde x,h)\, d\tilde x \, dh  \\
&=\int\int
\left|\int 
[\|x-x_k\|^\gamma\tilde \psi_{j,k}(x-x_k)
+\bar m_{\theta,j}(\theta^*(r),x) ra]
k((x-\tilde x)/h)f_X(x)\, dx\right|_{-}^p
f_\mu(\tilde x,h)\, d\tilde x \, dh
\end{align*}
where the integrals are taken over $\|\tilde x-x_k\|<C r^{1/\gamma},h<C r^{1/\gamma}$ and $\theta^*(r)$ is between $\theta_0$ and $\theta_0+ra$ (we suppress the dependence of $\theta^*(r)$ on $x$ in the notation).
Using the change of variables
$u=(x-x_k)/r^{1/\gamma}$, $v=(x-x_k)/r^{1/\gamma}$, $\tilde h=h/r^{1/\gamma}$, this is equal to
\begin{align*}
&\int\int \left|\int 
[\|r^{1/\gamma}u\|^\gamma\tilde \psi_{j,k}(r^{1/\gamma}u)
+\bar m_{\theta,j}(\theta^*(r),x_k+r^{1/\gamma}u) ra]
k((u-v)/\tilde h)f_X(x_k+r^{1/\gamma}u)r^{d_X/\gamma}\, du\right|_{-}^p  \\
&f_\mu(x_k+r^{1/\gamma}v,r^{1/\gamma}\tilde h)r^{d_X/\gamma}\, dv r^{1/\gamma}\, d\tilde h  \\
&=%
r^{[d_X+1+p(\gamma+d_X)]/\gamma}
\int\int \left|\int 
[\|u\|^\gamma\tilde \psi_{j,k}(r^{1/\gamma}u)
+\bar m_{\theta,j}(\theta^*(r),x_k+r^{1/\gamma}u) a]
k((u-v)/\tilde h)f_X(x_k+r^{1/\gamma}u)\, du\right|_{-}^p  \\
&f_\mu(x_k+r^{1/\gamma}v,r^{1/\gamma}\tilde h)\, dv \, d\tilde h
\end{align*}
where the integrals are taken over
$\|v\|<C,\tilde h<C$.  The result now follows from the dominated convergence theorem (here, and in subsequent results involving sequences of the form $\int|\int g_n(z,w)\, d\mu(z)|_{-}^p\, d\nu(w)$, the dominated convergence theorem is applied to the inner integral for each $w$, and again to the outer integral).

\end{proof}

\begin{lemma}\label{drift_approx_var_lemma}
Under the conditions of Theorem \ref{var_weight_lim_thm},
for any $a\in\mathbb{R}^{d_\theta}$,
\begin{align*}
&r^{-[d_X+p(d_X/2+\gamma)+1]/\gamma}
\int\int \sum_{j=1}^{d_Y}|E m_j(W_i,\theta_0+ra)k((X_i-\tilde x)/h)
  /(\sigma_j(\theta_0+ra,\tilde x, h)\vee \sigma_n)|_{-}^p
    f_\mu(\tilde x,h)\, d\tilde x \, dh  \\
&\le \sum_{k=1}^{\mathcal{X}_0} \sum_{j\in \tilde J(k)}
  \lambda_{\text{var}}(a,j,k,p) + o(1)
\end{align*}
for any $r=r_n\to 0$.  If, in addition,
$\sigma_n r_n^{-d_X/(2\gamma)}\to 0$,
the above display
will hold with the inequality replaced by equality.
\end{lemma}
\begin{proof}
As in the previous lemma, the following argument assumes, for simplicity, that $\gamma(j,k)=\gamma$ for all $(j,k)$ with $j\in J(k)$.
Let $\tilde s_j(r,\tilde x,h)=\sigma_j(\theta_0+ra,\tilde x,h)/h^{d_X/2}$.
As before, for large enough $C$, the integrand will be zero unless $\max\{\|\tilde x-x_k\|,h\}<C r^{1/\gamma}$ for some $k$ with $j\in J(k)$.  Thus, it suffices to prove the result for, fixing $(j,k)$ with $j\in J(k)$,
\begin{align*}
&\int\int |E m_j(W_i,\theta_0+ra)k((X_i-\tilde x)/h)
  (h^{-d_X/2}\tilde s_j^{-1}(r,\tilde x,h)\wedge \sigma_n^{-1})|_{-}^p
    f_\mu(\tilde x,h)\, d\tilde x \, dh  \\
&=\int\int 
  \left|\int [\|x-x_k\|^\gamma \tilde \psi_{j,k}(x-x_k)
   +\bar m_{\theta,j}(\theta^*(r),x) ra]  \right.  \\
&\left. k((x-\tilde x)/h)
     (h^{-d_X/2}\tilde s_j^{-1}(r,\tilde x,h)\wedge \sigma_n^{-1})f_X(x)\, dx
  \right|_{-}^p
    f_\mu(\tilde x,h)\, d\tilde x \, dh
\end{align*}
where the integral is taken over $\|\tilde x-x_k\|<C r^{1/\gamma}$, $h<C r^{1/\gamma}$ and $\theta^*(r)$ is between $\theta_0$ and $\theta_0+ra$.  Using the change of variables
$u=(x-x_k)/r^{1/\gamma}$, $v=(\tilde x-x_k)/r^{1/\gamma}$,$\tilde h=h/r^{1/\gamma}$, this is equal to
\begin{align*}
&\int\int
  \left|\int r[\|u\|^\gamma \tilde \psi_{j,k}(r^{1/\gamma}u)+\bar m_{\theta,j}(\theta^*(r),x_k+ur^{1/\gamma}) a]k((u-v)/\tilde h)\right.  \\
&\left.    (((r^{1/\gamma}\tilde h)^{-d_X/2}\tilde s_j^{-1}(r,x_k+vr^{1/\gamma},r^{1/\gamma}\tilde h))\wedge \sigma_n^{-1})
f_X(x_k+ur^{1/\gamma})r^{d_X/\gamma}\, du\bigg.\right|_{-}^p  \\
& f_\mu(x_k+vr^{1/\gamma},r^{1/\gamma}\tilde h)r^{d_X/\gamma}\, dv r^{1/\gamma}\, d\tilde h  \\
&=%
r^{[p(\gamma+d_X/2)+d_X+1]/\gamma}\int\int
  \left|\int [\|u\|^\gamma \tilde \psi_{j,k}(r^{1/\gamma}u)+\bar m_{\theta,j}(\theta^*(r),x_k+ur^{1/\gamma}) a]k((u-v)/\tilde h)\right.  \\
&\left.((\tilde h^{-d_X/2}\tilde s_j^{-1}(r,x_k+vr^{1/\gamma},r^{1/\gamma}\tilde h))\wedge (r^{d_X/(2\gamma)}\sigma_n^{-1}))  
f_X(x_k+ur^{1/\gamma})\, du\bigg.\right|_{-}^p
    f_\mu(x_k+vr^{1/\gamma},r^{1/\gamma}\tilde h)\, dv \, d\tilde h.
\end{align*}
where the integral is taken over $\|v\|<C$, $h<C$.
By Lemma \ref{sigma_lim_lemma} and the dominated convergence theorem, this converges to $\lambda_{var}(a,j,k,p)$ if $\sigma_n r_n^{-d_X/(2\gamma)}\to 0$.  If $\sigma_n r_n^{-d_X/(2\gamma)}$ does not converge to zero, the above display is bounded from above by the same expression with $\sigma_n^{-1}$ replaced by $\infty$.

\end{proof}

\begin{lemma}\label{drift_approx_kern_lemma}
Under the conditions of Theorem \ref{kern_lim_thm},
for any $a\in\mathbb{R}^{d_\theta}$,
\begin{align*}
&r^{-(\gamma p+d_X)/\gamma}
\int \sum_{j=1}^{d_Y} \left|
  [Em_j(W_i,\theta_0+ra)k((X_i-x)/h)/Ek((X_i-x)/h)]
\omega_j(\theta_0+ra,x)
   \right|_{-}^p \, d x  \\
&\to \sum_{k=1}^{|\mathcal{X}_0|}\sum_{j\in J(k)}
 \lambda_{\text{kern}}(a,c_{h,r},j,k,p)
\end{align*}
as $r\to 0$ with $h/r^{1/\gamma}\to c_{h,r}$ for $c_{h,r}>0$.  If the limit is zero for $(a,c_{h,r})$ in a neighborhood of the given values, the sequence will be exactly equal to zero for large enough r.

If $h/r^{1/\gamma}\to 0$, then, as
$r\to 0$,
\begin{align*}
&r^{-(\gamma p+d_X)/\gamma}
\int \sum_{j=1}^{d_Y} \left|
  [Em_j(W_i,\theta_0+ra)k((X_i-x)/h)/Ek((X_i-x)/h)]
\omega_j(\theta_0+ra,x)
   \right|_{-}^p \, d x  \\
&\to \sum_{k=1}^{|\mathcal{X}_0|}\sum_{j\in J(k)}
 \tilde\lambda_{\text{kern}}(a,j,k,p).
\end{align*}

\end{lemma}
\begin{proof}
As before, this proof treats the case where $J(k)=\tilde J(k)$ for ease of exposition.  As with the proofs of Lemmas \ref{drift_approx_bdd_lemma} and \ref{drift_approx_var_lemma}, it suffices to prove the result for, fixing $(j,k)$ with $j\in J(k)$,
\begin{align*}
&\int \left|
  [Em_j(W_i,\theta_0+ra)k((X_i-\tilde x)/h)/Ek((X_i-\tilde x)/h)]
\omega_j(\theta_0+ra,\tilde x)
   \right|_{-}^p \, d \tilde x  \\
&=\int \left|
  \int [\|x-x_k\|^\gamma \tilde \psi_{j,k}(x-x_k)
   +\bar m_{\theta,j}(\theta^*(r),x) ra]k((x-\tilde x)/h)f_X(x)\, dx 
   h^{-d_X}b(\tilde x)
\omega_j(\theta_0+ra,\tilde x)
   \right|_{-}^p \, d \tilde x
\end{align*}
where the integral is over $\|\tilde x-x_k\|< C r^{1/\gamma}$
and $b(\tilde x)\equiv h^{d_X}/Ek((X_i-\tilde x)/h)$ converges to
$(f_X(x_k))^{-1}$ uniformly over $\tilde x$ in any shrinking neighborhood of $x_k$ by Lemma \ref{k_lim_lemma}.  Let $\tilde h=h/r^{1/\gamma}$.  By the change of variables $u=(x-x_k)/r^{1/\gamma}$, $v=(\tilde x-x_k)/r^{1/\gamma}$, the above display is equal to
\begin{align}\label{kern_seq_eq}
&\int \left|
  \int [\|ur^{1/\gamma}\|^\gamma \tilde \psi_{j,k}(ur^{1/\gamma})
   +\bar m_{\theta,j}(\theta^*(r),x_k+ur^{1/\gamma}) ra]k((u-v)/\tilde h)f_X(x_k+ur^{1/\gamma})r^{d_X/\gamma}\, du
\right.  \nonumber  \\
&\left.   (r^{1/\gamma}\tilde h)^{-d_X}b(x_k+vr^{1/\gamma})
\omega_j(\theta_0+ra,x_k+r^{1/\gamma} v)
   \right|_{-}^p r^{d_X/\gamma}\, dv  \nonumber  \\
&=r^{p+d_X/\gamma}
\int \left|
  \int [\|u\|^\gamma \tilde \psi_{j,k}(ur^{1/\gamma})
   +\bar m_{\theta,j}(\theta^*(r),x_k+ur^{1/\gamma}) a]k((u-v)/\tilde h)f_X(x_k+ur^{1/\gamma})\, du
\right.  \nonumber  \\
&\left.   \tilde h^{-d_X}b(x_k+vr^{1/\gamma})
\omega_j(\theta_0+ra,x_k+r^{1/\gamma} v)
   \right|_{-}^p \, dv
\end{align}
where the integral is over $v<C$.
The first display of the lemma (the case where $h/r^{1/\gamma}\to c_{h,r}$ for $c_{h,r}>0$) follows from this and the dominated convergence theorem.

To show that the sequence is exactly zero for small enough $r$ when the limit is zero in a neighborhood of $(a,c_{h,r})$, note, that, if the limit is zero in a neighborhood of $(a,c_{h,r})$, we will have, for all $(\tilde a,\tilde c_{h,r})$ in this neighborhood and any $v$,
\begin{align*}
&\int \left[\|u\|^\gamma \psi_{j,k}\left(\frac{u}{\|u\|}\right)
   +\bar m_{\theta,j}(\theta_0,x_k) \tilde a\right]k((u-v)/\tilde c_{h,r})
 \, du  \\
&=\int \left[\tilde c_{h,r}^{\gamma}\|\tilde u \|^\gamma \psi_{j,k}\left(\frac{u}{\|u\|}\right)
   +\bar m_{\theta,j}(\theta_0,x_k) \tilde a\right]k(\tilde u-\tilde v)
 \, \tilde c_{h,r}^{d_X} d\tilde u
\ge 0.
\end{align*}
Evaluating this at $(\tilde c_{r,h},\tilde a)$
such that $\tilde c_{h,r}^{\gamma}\le c_{h,r}^{\gamma}(1-\varepsilon)$ and
(for the case where $\bar m_{\theta,j}(\theta_0,x_k) a$ is negative)
$\bar m_{\theta,j}(\theta_0,x_k) \tilde a
  \le (\bar m_{\theta,j}(\theta_0,x_k) a)(1+\varepsilon)$
shows that
\begin{align*}
\int \left[c_{h,r}^{\gamma}\|\tilde u \|^\gamma \psi_{j,k}\left(\frac{u}{\|u\|}\right)
  \cdot (1-\varepsilon)
   +(\bar m_{\theta,j}(\theta_0,x_k) a)(1+\varepsilon)\right]k(\tilde u-\tilde v)
 \,  d\tilde u
\ge 0
\end{align*}
for all $v$ for some $\varepsilon>0$.  The above display is, for small enough $r$, a lower bound for the inner integral in (\ref{kern_seq_eq}) times a constant that does not depend on $r$, so that, for small enough $r$, the inner integral in (\ref{kern_seq_eq}) will be nonnegative for all $v$ and (\ref{kern_seq_eq}) will eventually be equal to zero.

For the case where $\tilde h=h/r^{1/\gamma}\to 0$, multiplying (\ref{kern_seq_eq}) by $r^{-(p+d_X/\gamma)}$ gives, after the change of variables
$\tilde u=(u-v)/\tilde h$,
\begin{align*}
&\int \left|
  \int [\|\tilde h \tilde u+v\|^\gamma 
 \tilde \psi_{j,k}((\tilde h \tilde u+v)r^{1/\gamma})
   +\bar m_{\theta,j}(\theta^*(r),x_k+(\tilde h\tilde u+v)r^{1/\gamma}) a]k(\tilde u)
 f_X(x_k+(\tilde u \tilde h+v)r^{1/\gamma})\, d\tilde u
\right.  \\
&\left. b(x_k+vr^{1/\gamma})
\omega_j(\theta_0+ra,x_k+r^{1/\gamma} v)
   \right|_{-}^p \, dv
\end{align*}
which converges to
\begin{align*}
\int \left|[\|v\|^\gamma \psi_{j,k}(v/\|v\|)+\bar m_{\theta,j}(\theta_0,x_k) a]
\omega_j(\theta_0,x_k) \right|_{-}^p\, dv
\end{align*}
by the dominated convergence theorem, as required.

\end{proof}

We are now ready for the proofs of the main results.

\begin{proof}[proof of Theorem \ref{bdd_weight_lim_thm}]
The result follows immediately from Lemmas 
\ref{inst_equiv_lemma_2} and
\ref{drift_approx_bdd_lemma} since \linebreak
$(n^{-\gamma/\{2[d_X+\gamma+(d_X+1)/p]\}})^{-[d_X+p(d_X+\gamma)+1]/(\gamma p)}
=n^{1/2}$.
\end{proof}

\begin{proof}[proof of Theorem \ref{var_weight_lim_thm}]
The result follows immediately from Lemmas 
\ref{equiv_lemma_1}, \ref{inst_equiv_lemma_2} and
\ref{drift_approx_var_lemma} since \linebreak
$(n^{-\gamma/\{2[d_X/2+\gamma+(d_X+1)/p]\}})^{-[d_X+p(d_X/2+\gamma)+1]/(\gamma p)}
=n^{1/2}$.
\end{proof}

\begin{proof}[proof of Theorem \ref{kern_lim_thm}]
The result follows from Lemmas 
\ref{equiv_lemma_1}, \ref{kern_equiv_lemma_2}
and \ref{drift_approx_kern_lemma}.  Note that $(nh^{d_X})^{p/2}/(n^{1-d_X s})^{p/2}\stackrel{p}{\to} c_h^{d_X p/2}$, and that, for the case where
$s\ge 1/[2(\gamma+d_X/p+d_X/2)$,
\begin{align*}
(n^{-q})^{-(\gamma p+d_X)/(\gamma p)}
  =(n^{-(1-s d_X)/[2(1+d_X/(p\gamma))]})^{-(\gamma p+d_X)/(\gamma p)}
  =n^{(1-s d_X)/2}.
\end{align*}
For the case where $s< 1/[2(\gamma+d_X/p+d_X/2)]$, it follows from Lemmas 
\ref{equiv_lemma_1}, \ref{kern_equiv_lemma_2}
and \ref{drift_approx_kern_lemma} that
\begin{align*}
&n^{q(\gamma p+d_X)/(\gamma p)} T_n(\theta_0+a_n)  
\stackrel{p}{\to} \left(\sum_{k=1}^{|\mathcal{X}_0|}\sum_{j\in J(k)}
     \lambda_{\text{kern}}(a,c_h,j,k,p)\right)^{1/p}
\end{align*}
so that $(nh^{d_X})^{1/2}T_n(\theta_0+a_n)$ will converge to $\infty$ in this case if the limit in the above display is strictly positive.  If the limit in the above display is zero in a neighborhood of $(a,c_h)$, it follows from 
Lemmas \ref{equiv_lemma_1} and \ref{kern_equiv_lemma_2}
that $(nh^{d_X})^{1/2}T_n(\theta_0+a_n)$ is, up to $o_p(1)$, equal to a term that is zero for large enough $n$ by Lemma \ref{drift_approx_kern_lemma}.

\end{proof}

\end{document}